\documentclass[onecolumn,aps,prd,a4paper,superscriptaddress,tightenlines,noeprint,11pt,floatfix]{revtex4-2}

\usepackage{silence}
\usepackage[utf8]{inputenc}
\usepackage[USenglish]{babel}

\usepackage[T1]{fontenc}
\usepackage{lmodern}

\usepackage{mathtools, amsthm, amssymb, bbm}
\usepackage{braket}
\usepackage{bm}
\usepackage{soul}

\usepackage{microtype}
\usepackage{enumitem}
\usepackage[normalem]{ulem}
\usepackage{thm-restate}

\usepackage{graphicx}
\usepackage{float}

\usepackage[toc,page,header]{appendix}
\usepackage{minitoc}

\usepackage[bookmarksnumbered,colorlinks=true,linkcolor=blue,urlcolor=black,citecolor=blue,anchorcolor=green]{hyperref}
\usepackage[capitalize]{cleveref}

\allowdisplaybreaks
\numberwithin{equation}{section}
\newcommand{\B}{\textbf}
\newcommand{\be}{\begin{equation}}
\newcommand{\ee}{\end{equation}}
\newcommand{\ba}{\begin{eqnarray}}

\newcommand{\ea}{\end{eqnarray}}
\newcommand{\ot}{\otimes}

\newcommand{\trace}{\mathrm{Tr}}
\newcommand{\absval}[1]{\left| {#1} \right|}

\newcommand{\exv}{\mathbb{E}}
\newcommand{\lin}{\operatorname{lin}}

\newcommand{\wg}{\text{W\!g}}

\newcommand{\cmmnt}[1]{}

\newtheorem{theorem}{Theorem}

\begin{document}
\setcounter{secnumdepth}{3}

\title{Stabilizer Entropy of Subspaces}

\author{Simone Cepollaro}
\email{simone.cepollaro-ssm@unina.it}
\affiliation{Scuola Superiore Meridionale,  Largo S. Marcellino 10, 80138 Napoli, Italy}
\affiliation{Department of Mechanical Engineering, Massachusetts Institute of Technology, 77 Massachusetts Ave, Cambridge, MA 02139, USA}
\affiliation{INFN, Sezione di Napoli, Italy}

\author{Gianluca Cuffaro}\email{gianluca.cuffaro001@umb.edu}
\affiliation{Physics Department,  University of Massachusetts Boston,  02125, USA}

\author{Matthew B. Weiss}\email{matthew.weiss001@umb.edu }
\affiliation{Physics Department,  University of Massachusetts Boston,  02125, USA}

 \author{Stefano Cusumano}
 \affiliation{INFN, Sezione di Napoli, Italy}
 \affiliation{Dipartimento di Fisica `Ettore Pancini', Universit\`a degli Studi di Napoli Federico II,
 Via Cintia 80126,  Napoli, Italy}

 \author{Alioscia Hamma}
 \affiliation{Scuola Superiore Meridionale,  Largo S. Marcellino 10, 80138 Napoli, Italy}
 \affiliation{INFN, Sezione di Napoli, Italy}
 \affiliation{Dipartimento di Fisica `Ettore Pancini', Universit\`a degli Studi di Napoli Federico II,
 Via Cintia 80126,  Napoli, Italy}

\author{Seth Lloyd}
\affiliation{Department of Mechanical Engineering, Massachusetts Institute of Technology, 77 Massachusetts Ave, Cambridge, MA 02139, USA}

\begin{abstract}
We consider the costs and benefits of embedding the states of one quantum system within those of another.  Such embeddings are ubiquitous, e.g., in error correcting codes and in symmetry-constrained systems. In particular we investigate the impact of embeddings in terms of the resource theory of nonstabilizerness (also known as \emph{magic}) quantified via the \emph{stabilizer entropy} (SE). We analytically and numerically study the \emph{stabilizer entropy gap} or \emph{magic gap}: the average gap between the SE of a quantum state realized within a subspace of a larger system and the SE of the quantum state considered on its own.  We find that while the stabilizer entropy gap is typically positive, requiring the injection of magic, both zero and negative magic gaps are achievable. This suggests that certain choices of embedding subspace provide strong resource advantages over others. We provide formulas for the average nonstabilizerness of a subspace given its corresponding projector and sufficient conditions for realizing zero or negative gaps: in particular, certain classes of stabilizer codes provide paradigmatic examples of the latter. Through numerical optimization, we find subspaces which achieve both minimal and maximal average SE for a variety of dimensions, and compute the magic gap for specific error-correcting codes and symmetry-induced subspaces. Our results suggest that a judicious choice of embedding can lead to greater efficiency in both classical and quantum simulations.
\end{abstract}
\maketitle

\section{Introduction}

It is often profitable to embed the states and dynamics of one quantum system into the those of another. For example, any fault-tolerant quantum computer ought to employ error-correcting codes whereby logical degrees of freedom are encoded in a noise-resistant way within the physical degrees of freedom \cite{Nielsen_Chuang_2010}. Moreover, most current proposals for quantum computing architectures are based on qubits as the fundamental building block since implementing generic qudit gates and performing qudit measurement and read-out represents a significant challenge for current quantum hardware \cite{meth_simulating_2025,PRXQuantum.5.040309,PhysRevResearch.6.013050}. Thus if one wishes to simulate systems of higher dimensionality, one will generally need to encode such qudits in a particular subspace of the $n$-qubit Hilbert space. A physically relevant case is that of symmetries or gauge structures, which are imposed by a projection from a total Hilbert space to a subspace: for example, the low-energy sector of a quantum system, the subspace of states invariant under the symmetry of a Hamiltonian, or the gauge-invariant space in a lattice gauge theory \cite{zeng2019quantum,kitaev_anyons_2006,PhysRevA.72.012324}.

The choice of such a subspace is constrained by the problem considered. Depending on the architecture of a given quantum computer, a particular error correcting code might be especially natural to implement. In the context of quantum simulation, a choice of subspace may be picked out by the physics of the system under study. For instance, quantum computing \cite{lloyd1993potentially,cirac2000scalable} has recently emerged as a promising platform for the simulation of lattice gauge theories (LGT) \cite{santra2025quantumresourcesnonabelianlattice}, which provide a non-perturbative framework for studying quantum field theories  through a Hamiltonian formulation and are crucial tools for understanding strongly interacting systems, e.g., in quantum chromodynamics \cite{savage2016nuclearphysics,constantinou2015recentprogresshadronstructure,Meyer_2019}. Simulating LGTs requires exponential classical computational resources \cite{PhysRevResearch.5.043128,illa2025dynamicallocaltadpoleimprovementquantum,dimarcantonio2025rougheningdynamicselectricflux,PhysRevLett.131.171902,Xu:2025abo}, therefore recent works \cite{PhysRevD.108.074503,PhysRevD.102.074501,Robaina_2021,Bañuls:20198n,banuls_simulating_2020,nt76-ttmj,calliari2025fielddigitizationscalingmathbbzn,Joshi:2025rha,Joshi:2025pgv,tian2025roleplaquettetermgenuine} have proposed methods and ideas for real-time quantum LGT simulations. 

On the other hand, one may have some freedom to pick a subspace, and in this case, one may be guided by resource-theoretic considerations. In this paper, we will consider the interplay between the choice of encoded subspace and a particular quantum resource, \emph{nonstabilizerness}, commonly referred to as \emph{magic}, which quantifies how far a quantum state is from the set of stabilizer states \cite{Gottesman_98,Aaronson_04}. By the Gottesman-Knill theorem, the fragment of quantum theory consisting of stabilizer states, stabilizer-preserving (Clifford) operations, and stabilizer measurements is efficiently classically simulatable: consequently, the presence of magic strongly affects e.g., the number of samples necessary to estimate expectation values \cite{gottesman1998heisenbergrepresentationquantumcomputers, PhysRevLett.118.090501, PhysRevLett.115.070501}. At the same time, in the \emph{magic state model} of quantum computing \cite{Veitch2014}, stabilizer states and operations are taken to be ``free'' operations, and computational universality is achieved precisely by injecting a supply of nonstabilizer states. From a resource theory point of view, the cost of a quantum computation can be quantified in terms of the number of standardized magic states which must be injected. Given the central role of nonstabilizerness as a resource for universal quantum computation, a variety of measures of nonstabilizerness have been proposed in the literature \cite{PRXQuantum.4.010301, PhysRevLett.124.090505, Peleg_2022, Wang_2019}. We adopt the \emph{stabilizer entropy} (SE) \cite{Leone_2022}, an entropic measure which is the essentially unique computable measure of nonstabilizerness for pure states \cite{PhysRevA.110.L040403} and for which efficient numerical algorithms are available \cite{PhysRevA.109.L040401,PhysRevLett.132.240602}. The SE can be experimentally measured on quantum processors \cite{oliviero_measuring_2022,ahmad2025experimentaldemonstrationnonlocalmagic}, in particular via an efficient quantum algorithm \cite{Haug_2024}. The SE has highlighted nonstabilizerness as a crucial quantum resource, shedding light on its key role in the complex behavior of chaotic quantum systems \cite{jasser2025stabilizerentropyentanglementcomplexity,odavić2025stabilizerentropynonintegrablequantum}, allowing a comprehensive study of the interplay between entanglement and nonstabilizerness \cite{iannotti2025entanglementstabilizerentropiesrandom}, and deepening our understanding of Clauser–Horne–Shimony–Holt (CHSH) inequality violations \cite{cusumano2025nonstabilizernessviolationschshinequalities}.

The central question we consider here is: how does the nonstabilizerness required for quantum computation change as a function of  the choice of embedding subspace? Practically speaking, we try to answer the question: which choice of subspace is optimal from the point of view of the resource theory of magic? On the one hand, given a state in a Hilbert space of dimension $d_S$, one may calculate its stabilizer entropy with respect to a choice of Weyl-Heisenberg operators (in terms of which the stabilizer states themselves are defined): we shall call this its \emph{intrinsic} stabilizer entropy. On the other hand, we may reinterpret the same state as living a particular $d_S$-dimensional subspace of a larger $d_B$ dimensional Hilbert space and calculate its \emph{extrinsic} stabilizer entropy with respect to a choice of WH operators on the larger Hilbert space. The difference between the extrinsic and intrinsic stabilizer entropies we term the \emph{magic gap} of the embedding, a concept originally introduced in \cite{PhysRevD.109.126008,doi:10.1142/S0219887825500033}. A positive magic gap signifies that preparing the encoded system costs more than preparing the system on its own; zero magic gap means that there is no additional cost (at least in terms of magic) to preparing the encoded system; and most tantalizingly, a negative magic gap would suggest that one actually does better by simulating a system as a part of a larger whole than by treating it on its own. In order to isolate the effect of the choice of subspace, we investigate the \emph{average magic gap}: the intrinsic stabilizer entropy may be averaged (with respect to Haar measure) over all states in $\mathcal{H}_{d_S}$, and this can be compared to extrinsic stabilizer entropy averaged over all states in the chosen $d_S$-dimensional subspace of $\mathcal{H}_{d_B}$. The resulting gap diagnoses the resource cost associated to the choice of subspace itself. 

The paper is organized as follows: in \cref{II}, we review the stabilizer formalism for finite dimensional Hilbert spaces, characterizing stabilizer states and the Clifford group for qudits, and introduce the measures of nonstabilizerness employed throughout the work---namely, the stabilizer entropy and its linear variant. In \cref{ase_gap}, we calculate exactly the average intrinsic stabilizer entropy of an arbitrary Hilbert space in both multiqubit and multiqudit scenarios, and we provide an explicit, computationally tractable formula for the average extrinsic stabilizer entropy of a subspace expressed in terms of the characteristic function of the projector. \cref{overall} proves that the average magic gap itself averaged over all choices of subspace of fixed dimension is just the average magic of the larger Hilbert space. In \cref{overall} we report the results of an extensive numerical search for subspaces of minimal and maximal average magic for different choices of subspace dimension and larger Hilbert space dimension, and prove sufficient conditions for the average magic gap to be zero or negative in the case that the subspace corresponds to a stabilizer code (with trivial group homomorphism). In fact, we show that any such $n$-qubit stabilizer code which encodes $m$-qubits must have zero average magic gap, or negative average magic gap if the encoded system is viewed as a single qudit. Zero average magic gap is also always attainable if the overall dimension is odd, and under certain conditions, average magic gap $\leq 0$ is achievable in even qudit dimensions.

In \cref{complement}, we consider the possibility that a given state might attain lower stabilizer entropy by adding support in the complement of the chosen subspace, and calculate the average SE in the case that this complementary support is optimized for each state in the subspace, and in the case that a single support in the complement is chosen optimally for the given subspace. We provide numerical evidence that optimizing over a fixed support generally provides an advantage for a generic subspace, but not necessarily for a subspace which already achieves minimal average SE. Finally, in \cref{examples} we calculate the average magic gap for several interesting subspaces: the ground state space of a frustration-free Hamiltonian; the encoding of a spin-$j$ system in a state of $2j$-symmetrized qubits; the encoding of a qudit into an SU$(2)$ invariant subspace, relevant to the spin network construction of quantum gravity and SU$(2)$ gauge theories more generally; the invariant subspace of a $\mathbb{Z}_d$ gauge theory; and a [[4,2,2]] stabilizer code, which achieves zero magic gap. \cref{overall} provides a brief overview of the numerical methods used in this paper. In particular, we have developed an efficient open-source \texttt{python} package which implements all of the constructions in the paper, including exact and Monte Carlo approximations of the average magic gap.

We hope that the present work opens the door to a systematic resource-theoretic analysis of the choice of embeddings used in quantum computing and beyond. In particular, the fact that a negative magic gap is generically attainable for certain combinations of small and large Hilbert space dimensions raises some profound questions of practical importance. A negative magic gap implies that, on average, fewer resources are required to simulate a system through an embedding than to simulate the system in itself---at least by the lights of the resource theory of magic. Is the true cost borne elsewhere, e.g., in terms of the entanglement necessary, or in terms of the cost of the encoding and decoding maps? In what situations can one truly take advantage of this result? Moreover, an examination of the dimensions in which zero and negative magic gap are attainable raises interesting questions at the intersection of number theory and quantum information: we go partway towards resolving such questions, but the full story awaits its unfolding.

\section{Stabilizer Formalism and Stabilizer Entropy} \label{II}
\subsection{The Weyl-Heisenberg group}

For any natural number $d$, let $\mathbb{Z}_d=\{0,\dots,d-1\}$ be the set of integers modulo $d$. In order to construct a unitary representation of the group $\mathbb{Z}_d\times \mathbb{Z}_d$ on a $d$-dimensional Hilbert space $\mathbb{C}^d$,  we first define an orthonormal basis $\{\ket{k}\,|\,k\in\mathbb{Z}_{d}\}$ on which we let clock and shift operators act respectively as
\begin{align}
    Z\ket{k}&\equiv\omega^k\ket{k}\\
    X\ket{k}&\equiv\ket{k+1 \text{ mod } d},
\end{align}
where $\omega=e^{2\pi i /d}$. We note the clock and shift operators satisfy
\begin{equation}
    X^d=Z^d=I\quad\quad\quad\quad X^kZ^l=\omega^{-kl}Z^lX^k.
\end{equation}
For each pair of integer $\textbf{a}=(a_1,a_2)\in\mathbb{Z}_d\times\mathbb{Z}_d$ we may define the displacement operators as
\begin{equation}
    D_\textbf{a}\equiv \tau^{a_1 a_2}X^{a_1}Z^{a_2},
\end{equation}
where $\tau=-e^{i\pi/d}$. Notice that $\tau^d = 1$ if $d$ is odd, but $\tau^d=-1$ if $d$ is even: regardless, $\tau^{2d}=1$. Thus an important subtlety arises in even dimensions, where the phase factor has a $2d$ periodicity, and so indices to the displacement operators ought to be taken mod $2d$ instead of mod $d$. 

Using the properties of the clock and shift operators, it is straightforward to check that the displacement operators satisfy a number of important properties \cite{appleby2005}: 
\begin{align}
    &D_{\textbf{a}}^\dagger=D_{-\textbf{a}},\label{dagger}\\
    &D_\textbf{a}D_\textbf{b}=\tau^{[\textbf{a,\textbf{b}}]}D_{\textbf{a}+\textbf{b}}\label{HEalgebra}\\
    &\trace{(D_\textbf{a}^\dagger D_\textbf{b})}=d\,\delta_{\textbf{ab}}\label{orth},\\
    &D_{\B{a}}=D_{\textbf{x}+d\textbf{y}}=\begin{dcases}
        D_\textbf{x} &\text{$d$ odd}\\ (-1)^{[\textbf{x},\textbf{y}]}D_\textbf{x}\quad &\text{$d$ even,}
        \end{dcases}\label{evenodd}
\end{align}
where $[\textbf{a},\textbf{b}]=a_2 b_1 - a_1 b_2$ is a symplectic form defined on $\mathbb{Z}_d\times\mathbb{Z}_d$.   Eq.\! \ref{orth} establishes that the displacement operators form an orthogonal basis of unitary operators, and in interpreting Eq.\! \ref{evenodd}, we note that we may always decompose any symplectic index as $\B{a}=\B{x}+d\B{y}$ where $\B{x}=\B{a} \text{ mod } d$ and $\B{y}=\frac{1}{d}(\B{a}-\B{x})$, which provides an alternative to considering indices mod $2d$. (Intuitively, $\B{x}$ tells us how many ``minutes'' to go around each clock, and $\B{y}$ tells us how many ``hours'' have passed.)

The set of operators $\{D_\textbf{a}\}$, together with a phase factor $\omega^s$, generate the single qudit Weyl-Heisenberg (WH) group $\mathcal{W}_d$ 
\begin{equation}
    \mathcal{W}_d\equiv\{\omega^sD_\textbf{a}\}\,,\quad\quad s\in\mathbb{Z},\,\,\textbf{a}\in\mathbb{Z}_d^2.
\end{equation} 
The same structure can be easily extended to multiqudit systems. We take as our Hilbert space $(\mathcal{H}_d)^{\otimes n}$, where  $d$ is the local dimension, and $n$ is the number of qudits. We consider symplectic indices valued in $\mathbb{Z}_d^{2n}$, that is, $\B{a}=\B{a}_1 \oplus \dots \oplus \B{a}_k$, and the symplectic product becomes $[\B{a}, \B{b}]=\sum_i [\B{a}_i, \B{b}_i]$. The phase factors are as before $\omega=e^{2\pi i/d}$ and $\tau=-e^{\pi i/d}$. Displacement operators acting on $\mathcal{H}_{d}^{\otimes n}$ may be defined as the tensor product of displacements on each subsystem~\cite{gross, Veitch_2012}, and so the multiqudit Weyl-Heisenberg group is therefore 
\begin{equation}
    \mathcal{W}_d^n\equiv\{\omega^sD_{\textbf{a}_1}\otimes\dots\otimes D_{\textbf{a}_n}\}\,,\quad\quad s\in\mathbb{Z},\,\,\textbf{a}_i\in\mathbb{Z}_d^{2n}.
\end{equation}

\subsection{Stabilizer subspaces and the Clifford group}
\label{stab_subspaces}

Having introduced the Weyl-Heisenberg displacement operators, we may define stabilizer codes, and in particular the stabilizer states, along with the Clifford group which preserves them. Let $d$ be the local dimension, and $n$ the number of qudits, equipped with a set of WH operators $\{D_{\B{a}}\}$. A \emph{stabilizer code} is a subspace of $(\mathbb{C}^d)^{\otimes n}$ which corresponds to the common $+1$ eigenspace of an abelian subgroup of the WH group called the stabilizer group of the code. If the eigenspace is in fact 1-dimensional, the corresponding pure state is called a \emph{stabilizer state}.

To better understand this construction, consider a set of symplectic indices $\mathcal{S}$. If $\forall \B{a}, \B{b} \in \mathcal{S}, [\B{a}, \B{b}]=0$, then $D_\B{a}D_\B{b}=\tau^{[\B{a}, \B{b}]}D_{\B{a}+\B{b}}=D_\B{a+b}=D_\B{b}D_\B{a}$, that is, the corresponding Weyl-Heisenberg operators commute. If we also demand  that the set is closed under addition and contains the zero element, then $\mathcal{S}$ forms a subgroup of $\mathbb{Z}_d^{2n}$ called \emph{totally isotropic}, and the corresponding WH operators form an abelian subgroup of the WH group. A simple example is the set $\mathcal{S} = \Big\{(p, 0) \ | \ p\in \mathbb{Z}_{d}, p=0 \text{ mod } d_S\Big\}$,
where we've taken $n=1$, and supposed that $d_S$ divides $d$. Then $|\mathcal{S}|=d^n/d_S$, and for any two elements $\B{a}=(p,0)$ and $\B{b}=(p^\prime, 0)$, the symplectic form  $[\B{a}, \B{b}]=p (0) - (0)p^\prime=0$ clearly vanishes: moreover, the set is closed under addition, and contains 0. 

Alongside a totally isotropic set $\mathcal{S}$, let also consider a function $f(\B{a}): \mathcal{S}\rightarrow \mathbb{Z}_d$ which is a group homomorphism, that is, which satisfies $f(\B{a}+\B{b})=f(\B{a})+f(\B{b}) \text{ mod } d$. A stabilizer codespace is completely determined by an isotropic set $\mathcal{S}$ and a group homomorphism $f(\B{a})$ \cite{Gross2021}. We may construct the projector onto the codespace by summing over the group elements
\begin{align}
\Pi_{\mathcal{S},f} = \frac{1}{|S|}\sum_{\B{a}\in \mathcal{S}} \omega^{f(\B{a})}D_{\B{a}}.
\end{align}
We require $\Pi_{\mathcal{S},f}=\Pi_{\mathcal{S},f}^2$ so that
\begin{align}
 \Pi_{\mathcal{S},f}^2 &= \left[\frac{1}{|S|}\sum_{\B{a}\in \mathcal{S}} \omega^{f(\B{a})}D_{\B{a}}\right]\left[\frac{1}{|S|}\sum_{\B{b}\in \mathcal{S}} \omega^{f(\B{b})}D_{\B{b}}\right]=\frac{1}{|\mathcal{S}|^2}\sum_{\B{a},\B{b}\in \mathcal{S}}\omega^{f(\B{a})+f(\B{b})}\tau^{[\B{a}, \B{b}]}D_{\B{a}+\B{b}},
\end{align}
but  $\forall \B{a}, \B{b} \in \mathcal{S}, [\B{a}, \B{b}]=0$, so that letting $\B{c}=\B{a}+\B{b} \in \mathcal{S}$ and using the group homomorphism property of $f(\B{a})$,
\begin{align}
	 \Pi_{\mathcal{S},f}^2 &= \frac{1}{|\mathcal{S}|^2}\sum_{\B{a},\B{c}\in \mathcal{S}}\omega^{f(\B{a})+f(\B{c}-\B{a})}D_{\B{c}}=\frac{1}{|\mathcal{S}|}\sum_{\B{c}\in \mathcal{S}}\omega^{f(\B{c})}D_\B{c}=\Pi_{\mathcal{S},f},
\end{align}
as desired. Similarly, one may check that $\Pi_{\mathcal{S}, f}=\Pi_{\mathcal{S}, f}^\dagger$, and the corresponding codespace must have dimension $d_S=d^n/|\mathcal{S}|$ \cite{Gross2021}. 

Finally, when $|\mathcal{S}|=d^n$, $\Pi_{\mathcal{S},f}$ will be a rank-1 projector onto a so-called \emph{stabilizer state}. The \emph{Clifford group} $\mathcal{C}_d^n$ is defined as the normalizer of the WH group, and consequently preserves the set of stabilizer states. It consists of unitaries $U$ such that for any $D_\textbf{a}\in\mathcal{W}_d^n$,
\begin{equation}
    U D_\textbf{a}U^\dagger =e^{i\gamma}D_{\textbf{a}'},
\end{equation}
for $\textbf{a}'\in\mathbb{Z}_d^{2n}$ and $\gamma\in\mathbb{R}$.
A straightforward calculation shows that such unitaries indeed form a group.

\subsection{Stabilizer entropy}

We now introduce the particular measure of nonstabilizerness which we will employ throughout: the \emph{stabilizer entropy}. Since the WH operators form an orthogonal basis for $\mathcal{L}(\mathcal{H}_d^{\otimes n})$, we can expand a state $\psi$ in terms of them,
\begin{equation}
    \psi=\frac{1}{d^n}\sum_{\textbf{a}\in\mathbb{Z}_d^{2n}}\trace{(D_\textbf{a}^\dagger\psi)}\,D_\textbf{a}.
\end{equation}
We will call $\chi_\textbf{a}(\psi)\equiv\frac{1}{d^n}\trace{(D_\textbf{a}^\dagger\psi)}$ the characteristic function of the state $\psi$~\cite{gross,Veitch_2012}. Since the $D_\textbf{a}$'s are not Hermitian operators (except for $d=2$), the characteristic function will generally be complex-valued. However, by taking the absolute value squared of $\chi_\textbf{a}(\psi)$ and appropriately rescaling it,  we find that 
\begin{equation}
    P_\textbf{a}(\psi)\equiv\frac{1}{d^n}\absval{\trace{(D_\textbf{a}^\dagger\psi)}}^2
\end{equation}
forms a probability vector in the case that $\psi =|\psi\rangle\langle \psi|$ is pure.  We can then define the \textit{stabilizer R\'enyi entropy}~\cite{Leone_2022,wang_stabilizer_2023} of order $\alpha$  for a pure state as
\begin{equation}
    M_\alpha(\psi)\equiv\frac{1}{1-\alpha}\log{\sum_{\textbf{a}\in\mathbb{Z}_d^{2n}}P_\textbf{a}(\psi)^\alpha}-\log{d^n},
\end{equation}
where the $\log d^n$ offset ensures that this quantity is lower-bounded by zero. This definition may be extended to mixed states, but we will not need to do so for the sequel.

Stabilizer R\'enyi entropies have the following properties which make them good measures of nonstabilizerness from a resource-theory perspective~\cite{wang_stabilizer_2023, QRT}:
\begin{itemize}
    \item[(i)] Faithfulness: \begin{equation}
        M_\alpha(\psi)=0\quad \text{iff}\quad \psi\,\, \text{is a stabilizer state};
    \end{equation}
    \item[(ii)] Invariance under Clifford unitaries:  \begin{equation}
        M_\alpha(C\psi C^\dagger)=M_\alpha(\psi) \quad\quad \forall C\in\mathcal{C}_d^n;
    \end{equation}
    \item[(iii)] Additivity under tensor product:
    \begin{equation}
        M_\alpha(\psi\otimes\phi)=M_\alpha(\psi)+M_\alpha(\phi);
    \end{equation}
    \item[(iv)] Existence of a nontrivial upper bound:
    \begin{equation}\label{upb}
        M_\alpha(\psi)\leq\frac{1}{1-\alpha}\log{\frac{1+(d^n-1)(d^n+1)^{1-\alpha}}{d^n}},
    \end{equation}
    with equality if and only if $\psi$ belongs to a WH-covariant \emph{symmetric informationally complete set} \cite{cuffaro2024quantumstatesmaximalmagic}: this is just one of many extremality properties satisfied by SICs.
\end{itemize}
It is generally more tractable to deal with linearized stabilizer entropies based on the $\alpha$-Tsallis entropy \cite{Tsallis1988},
\begin{equation}
    M_{\lin}^{(\alpha)}(\psi)=\frac{1}{\alpha-1}\left(1-(d^{n})^{\alpha-1}\sum_{\textbf{a}}P_\textbf{a}(\psi)^\alpha\right).
\end{equation}
$M_{\lin}^{(\alpha)}$ is again faithful and invariant under Clifford unitaries for $\alpha\geq2$ \cite{PhysRevA.110.L040403}, and its linearity makes it particularly useful in practical computations, especially when dealing with averages. Indeed, notice that if we define \be Q_\alpha=\frac{1}{(d^{n})^\alpha}\sum_\textbf{a}(D_{\textbf{a}}\otimes D_{\textbf{a}}^\dagger)^{\otimes\alpha}, \ee we have compactly
\begin{equation}
    M_{\lin}^{(\alpha)}(\psi)=\frac{1}{\alpha-1}\biggr(1-(d^n)^{\alpha-1}\trace\left(Q_\alpha\,\psi^{\otimes2\alpha}\right)\biggr).
\end{equation}
Since $Q=Q^\dagger$, it is in fact an observable, which reassures us that $M_{\lin}^{(\alpha)}(\psi)$ is a real number. If $\mathcal{H}=(\mathbb{C}^2)^{\otimes n}$, the Weyl-Heisenberg group coincides with the standard Pauli group and $Q$ is in fact a projector: in general, however, $Q^2\neq Q$. In what follows, for simplicity, we will specialize to the linear stabilizer entropy with $\alpha=2$, so that
\begin{align}
    M(\psi)\equiv M_{\lin}^{(2)}(\psi)&=1-\frac{1}{d^{n}}\sum_{\B{a}}|\trace(D_{\B{a}}^\dagger \psi)|^4=1-d^n\trace\left(Q\,\psi^{\otimes4}\right),
\end{align}
for $Q\equiv Q_2$, where we have dropped the subscripts and superscripts to ease notation.

The stabilizer entropies can be directly related to the difficulty of both classical and quantum simulation. By the Gottesman-Knill theorem, a quantum process which consists of stabilizer operations on stabilizer states can be efficiently classically simulated. Working in the context of $n$-qubit systems, \cite{PhysRevLett.118.090501} provided a classical simulation algorithm for estimating expectation values of Pauli operators after stabilizer operations have been applied to an \emph{arbitrary state}, relating the number of samples required to a measure called the \emph{robustness of magic}, defined as
\begin{align}
\mathcal{R}(\psi) = \min_x\left\{\sum_i |x_i|, \psi=\sum_i x_i \sigma_i\right\}.
\end{align}
Here $\{\sigma_i\}$ is the set of stabilizer states, which form an overcomplete operator frame in which a state $\psi$ can be expanded and the $x_i$ are real such that $\sum_ix_i=1$ but not necessarily non-negative. Since the frame is overcomplete, there are different choices for expansion coefficients $x_i$: the robustness of magic is the least 1-norm which is achievable over choices of such vectors of coefficients, and can be viewed as a measure of negativity in this representation. The authors in  \cite{PhysRevLett.118.090501} establish that to estimate the expectation value within $\delta$ with probability greater than $1-\epsilon$, one must use $N$ samples where $N$ is
\begin{align}
N = \frac{2}{\delta^2}\left(\sum_i |x_i|\right)^2\ln\left(\frac{2}{\epsilon}\right) && \epsilon = 2 \exp\left(-\frac{N \delta^2}{2\left(\sum_i |x|_i\right)^2}\right)	.
\end{align}
In other words, the number of samples required for a Monte Carlo simulation scales as $O(\mathcal{R}(\psi)^2)$ \cite{Heinrich2019}. Comparable results may be obtained for the simulation of more general quantum systems, beyond qubits: the essence of the idea is that the cost of quantum Monte Carlo methods scales with the amount of negativity in the expansion coefficients $x_i$ \cite{PhysRevLett.115.070501}.

Crucially, the stabilizer entropies provide lower bounds on the robustness of magic \cite{Leone_2022, PhysRevLett.118.090501}. In particular, for $\alpha \ge \frac{1}{2}$, the R\'enyi stabilizer entropies (for pure states) satisfy
\begin{align}
M_\alpha(\psi) \leq 2 \log(\mathcal{R}(\psi)),
\end{align}
and for the 2-Tsallis stabilizer entropy in particular, we have (Appendix \ref{RobustnessBounds}),
\begin{align}
\mathcal{R}(\psi)^2\ge \frac{1}{1-M(\psi)}.	
\end{align}
Thus stabilizer entropies provide lower bounds on the robustness of magic, and consequently on the number of samples required in a classical simulation algorithm.

At the same time, the stabilizer entropy can be related to the difficulty of \emph{quantum} simulation. \cite{PhysRevA.110.L040403} has shown that that the $n$-qubit R\'enyi stabilizer entropies for $\alpha \ge 2$ are pure-state monotones, that is, $M_\alpha(\Phi(\psi))\leq M_\alpha(\psi)$ for any stabilizer operation $\Phi$, which includes not just Clifford unitaries, but also the operations of appending stabilizer ancillas, measuring in the computational basis, partial tracing and dephasing, as well as conditional operations on measurement outcomes or classical randomness.  The same applies to the linear stabilizer entropies based on the Tsallis entropy, and it is expected that this result holds as well for $n$-qudit stablizer entropies. For any pure-state monotone satisfying additivity under the tensor product, the rate of conversion from states $A$ to $B$ can be upper bounded by \cite{Beverland_2020}
\begin{align}
r_{A\rightarrow B} \leq \frac{M(A)}{M(B)}.
\end{align}
In other words, the maximum number of copies of $B$ that can be produced per copy of $A$ using only stabilizer operations is upper bounded by the ratio between their stabilizer entropies. In particular, if $M(A) \leq M(B)$, one cannot convert $A$ into $B$ using only stabilizer operations. Indeed, magic state distillation protocols provide a means to convert a collection of imperfect magic states into a small number of standard highly magic states: the stabilizer entropy ratio diagnoses the number of input states necessary to distill a particular magic state. A promising model for fault-tolerant quantum computation, the magic state model, proposes confining the computation to Clifford operations, which take stabilizer states to stabilizer states, along with the injection of some number of standard magic states \cite{Veitch2014}. Under the assumption that Clifford operations are trivial to perform, the stabilizer entropy helps to bound the resource cost in terms of the number of standard magic states necessary to perform an arbitrary quantum computation.

\section{Stabilizer Entropy gap} 
\label{ase_gap}
We now motivate and define the central quantity which we will investigate in this work: the \emph{average stabilizer entropy gap}, or more concisely, the \emph{average magic gap},  and provide explicit formulas for its evaluation. Suppose we have assigned a $d_S$ dimensional Hilbert space $\mathcal{H}_{d_S}$ to a quantum system, and we wish to simulate it by embedding it into a subspace of a larger Hilbert space $\mathcal{H}_{d_B}$ of dimension $d_B$. (We shall often refer to $\mathcal{H}_{d_S}$ as the ``small'' space, and $\mathcal{H}_{d_B}$ as the ``big'' space.) A choice naturally arises in how we should calculate the stabilizer entropy of a state: on the one hand, we could calculate the \emph{intrinsic} stabilizer entropy with respect to some choice of WH operators on $\mathcal{H}_{d_S}$; on the other hand, we could calculate the \emph{extrinsic} stabilizer entropy with respect to some choice of WH operators in $\mathcal{H}_{d_B}$.  In light of this, given a pure state $\psi \in \mathcal{L}(\mathcal{H}_{d_S})$, a linear embedding map $\mathcal{E}:\mathcal{H}_{d_S}\rightarrow \mathcal{H}_{d_B}$, and choices of WH operators $\{D^S_{\textbf{a}}\}$ and $\{D^B_{\textbf{a}}\}$ on the small and big spaces respectively, we define the \emph{stabilizer entropy gap of a state} as
\begin{align}
    \Delta M(\psi) &\equiv M_B(\mathcal{E}(\psi))-M_S(\psi)\\
    &=\Big(1 - d_B\trace\left(Q_B \mathcal{E}(\psi)^{\otimes 4}\right)\Big)-\Big(1 - d_S\trace\left(Q_S \psi^{\otimes 4}\right)\Big)\\
    &=d_S\trace\left(Q_S \psi^{\otimes 4}\right)-d_B\trace\left(Q_B \mathcal{E}(\psi)^{\otimes 4}\right),
\end{align}
where 
\begin{align}
  Q_S = \frac{1}{d_S^2}\sum_\textbf{a}(D^S_{\textbf{a}}\otimes D_{\textbf{a}}^{S\dagger})^{\otimes2} &&  Q_B = \frac{1}{d_B^2}\sum_\textbf{a}(D^B_{\textbf{a}}\otimes D_{\textbf{a}}^{B\dagger})^{\otimes2}.
\end{align}
The \emph{average stabilizer entropy gap} (ASE gap) $ \Delta M(\mathcal{E}) $ \cite{PhysRevD.109.126008,doi:10.1142/S0219887825500033} is this quantity averaged over all states in $\mathcal{H}_{d_S}$. The ASE gap therefore quantifies the resource disparity between a direct simulation of a $d_S$ dimensional quantum system and an indirect simulation of the same system living in particular $d_S$ dimensional subspace of a larger Hilbert space. Turning the Haar average over pure states in a Haar average over the unitary group $U(d_S)$, we have
\begin{align}
    \Delta M(\mathcal{E}) &\equiv \mathbb{E}_{\psi}\big[\Delta M(\psi)\big]\\
    &=\mathbb{E}_U[M_B(\mathcal{E}(U\psi U^\dagger))] -\mathbb{E}_U[M_S(U\psi U^\dagger)] \\
    &= d_S\trace\left(Q_S \mathcal{A}^S_4\right)-d_B\trace\left(Q_B \mathcal{E}^{\otimes 4}(\mathcal{A}^S_4)\right),
\end{align}
where
\begin{align}
   \mathcal{A}_4^S &=\mathbb{E}_U\big[(U\psi U^\dagger)^{\otimes 4}\big]= \binom{d_S+3}{4}^{-1}\Pi^S_{\text{sym}^4},
\end{align}
and 
\begin{align}
   \Pi^S_{\text{sym}^4} = \frac{1}{24}\sum_{\sigma\in S_4} T_{\sigma} &&T_{\sigma}=\sum_{i,j,k,l=0}^{d_S-1}\ket{\sigma(i,j,k,l)}\bra{ijkl}.
\end{align}
This last result derives from Schur-Weyl duality \cite{roberts_chaos_2017, Mele_2024}: the Haar average is proportional to the projector onto the permutation symmetric subspace of four tensor factors. 

On the one hand, through explicit calculation of $\trace(Q\Pi_{\text{sym}^4})$ in any dimension, \cref{App_Average} establishes that the average intrinsic stabilizer entropy is
\begin{align}
\mathbb{E}_U[M_S(U\psi U^\dagger)] &= 1 - d_S \binom{d_S+3}{4}^{-1}\trace\left( Q_S \Pi_{\text{sym}^4}^S\right)\\
&=\begin{dcases}
    1-\frac3{d_S+2}, &  d_S \text{ odd (multiqudit)} \\
    1-\frac{3(d_S+2)}{(d_S+1)(d_S+3)},   & d_S \text{ even (multiqudit)} \\
    1- \frac{4}{d_S+3}, & d_S=2^m \text{ (multiqubit)}.\\
\end{dcases}
\end{align}
The first two cases apply regardless of whether one employs the single qudit or multiqubit WH groups, $d_S$ being the total dimension. Only in the case that $d_S=2^m$, does the average stabilizer entropy depend on the choice of WH group: treated as a single qudit, one ought to use the second case; treated as $m$ qubits, one ought to use the third case. Indeed, this already shows that treating a system as $m$ qubits leads to lower ASE than treating a system as a single $2^m$-dimensional qudit, so that employing $m$ qubit stabilizer operations in one's computation is on average more efficient.

On the other hand, in order to evaluate $\mathbb{E}_U[M_B(\mathcal{E}(U\psi U^\dagger))]$, we make use of Lemma 1 from \cite{PhysRevD.109.126008} which assures us that, up to proportionality, we may trade the average over the subspace for an average over the larger Hilbert space,
\begin{align}
\mathbb{E}_U[M_B(\mathcal{E}(U\psi U^\dagger))]
&=  1 - d_B \trace\left( Q_S \mathcal{E}^{\otimes 4}(\mathcal{A}_4^S)\right)\\
&=1- d_B\binom{d_B+3}{4}\binom{d_S+3}{4}^{-1} \trace\left(Q_B\Pi_{\mathcal{E}}^{\otimes 4}\mathcal{A}^B_4 \right)\\
&=1-d_B \binom{d_S+3}{4}^{-1} \trace\left(Q_B\Pi_{\mathcal{E}}^{\otimes 4}\Pi_{\text{sym}^4}^B \right),
\end{align}
where $\mathcal{A}_4^B =\binom{d_B+3}{4}^{-1}\Pi^B_{\text{sym}^4}$ and $\Pi_{\mathcal{E}} \in \mathcal{L}(\mathcal{H}_{d_B})$ is the projector onto the desired $d_S$ dimensional subspace. This reformulation allows us to calculate the average extrinsic stabilizer entropy by working solely with operators defined on $\mathcal{H}_{d_B}$. \cref{App_Average} derives an explicit expression for $\trace(Q \Pi_{\mathcal{E}}^{\otimes 4} \Pi_{\text{sym}^4})$ in terms of the characteristic function $\chi_\B{a} := \frac{1}{d_B}\trace(D_\B{a}^\dagger \Pi_\mathcal{E})$ of the projector $\Pi_{\mathcal{E}}$,
\begin{align}
\label{bigguy}
\trace\left(Q\Pi_{\mathcal{E}}^{\otimes 4}\Pi_{\text{sym}^4} \right) &= \frac{1}{24}\Bigg\{3 d_B^2\sum_{\B{a}} |\chi_{\B{a}}|^4 + 6 d_B\sum_{\B{a}, \B{b}}  |\chi_{\B{a}}|^2  |\chi_{\B{b}}|^2 \omega^{[\B{a}, \B{b}]}+6 d_B\sum_{\B{a}, \B{b}}\chi_{\B{a}}^{2}\chi_{\B{b}}\chi_{\B{b}+2\B{a}}^*\\
&+8 \sum_{\B{a}, \B{b}, \B{c}}\chi_{\B{a}}\chi_{\B{b}}\chi_{\B{c}}\chi^*_{\B{a}+\B{b}+\B{c}}\tau^{[\B{a}, \B{b}]-[\B{a}, \B{c}]-[\B{c}, \B{b}]}+ \sum_{\B{a}}\left|\sum_{\B{b}}\chi_{\B{b}}\chi^*_{\B{b}-2\B{a}}\right|^2\Bigg\},\nonumber
\end{align}
where $\omega=e^{2\pi i/d_B}$, $\tau=-e^{i\pi/d_B}$  and $\left[\cdot,\cdot\right]$ denotes the symplectic product. In order to evaluate $\chi_\B{a}$ for an arbitrary symplectic index $\B{a}$, note that from Eq.\! \ref{evenodd}, the characteristic function satisfies for a $d$-dimensional Hilbert space
\begin{align}
\chi_{\B{a}}=\chi_{\B{x}+d\B{y}}=\begin{cases}
        \chi_{\B{x}} &\text{$d$ odd}\\ (-1)^{[\B{x},\B{y}]}\chi_{\B{x}}\quad &\text{$d$ even,}
        \end{cases}	
\end{align}
for $\B{x}=\B{a} \text{ mod } d$ and $\B{y}=\frac{1}{d}(\B{a}-\B{x})$. Putting this all together, the final expression of the ASE gap reads
\begin{align}
    \Delta M(\mathcal{E}) 
    &=-d_B \binom{d_S+3}{4}^{-1}  \frac{1}{24}\Bigg\{3 d_B^2\sum_{\B{a}} |\chi_{\B{a}}|^4 + 6 d_B\sum_{\B{a}, \B{b}}  |\chi_{\B{a}}|^2  |\chi_{\B{b}}|^2 \omega^{[\B{a}, \B{b}]}+6 d_B\sum_{\B{a}, \B{b}}\chi_{\B{a}}^{2}\chi_{\B{b}}\chi_{\B{b}+2\B{a}}^*\nonumber \\
    &+8 \sum_{\B{a}, \B{b}, \B{c}}\chi_{\B{a}}\chi_{\B{b}}\chi_{\B{c}}\chi^*_{\B{a}+\B{b}+\B{c}}\tau^{[\B{a}, \B{b}]-[\B{a}, \B{c}]-[\B{c}, \B{b}]} + \sum_{\B{a}}\left|\sum_{\B{b}}\chi_{\B{b}}\chi^*_{\B{b}-2\B{a}}\right|^2\Bigg\}\\&+\begin{dcases}
    \frac3{d_S+2}, &  d_S \text{ odd (multiqudit)} \\
    \frac{3(d_S+2)}{(d_S+1)(d_S+3)},   & d_S \text{ even (multiqudit)} \\
     \frac{4}{d_S+3}, & d_S=2^m \text{ (multiqubit)}.\\
    \end{dcases}\nonumber,
\end{align}

We now derive a specialized expression for the characteristic function $\chi_{\textbf{a}}$ of a projector onto an irreducible subspace of a representation of a group $G$, with particular emphasis on subspaces invariant under the group action. Consider a vector space $V$ on which a finite group $G$ acts via a representation $\rho: G \to GL(V)$. We can decompose $V$ into a direct sum of irreducible subspaces,
\be V = \bigoplus_i V_i \otimes \mathbb{C}^{\mu_i}, \ee
where $V_i$ are the irreducible subspaces and $\mu_i$ are their multiplicities. The projector onto the $i$-th sector (i.e., the subspace that transforms according to the irreducible representation $V_i$) is given by the formula
\be \Pi_i = \frac{\dim(V_i)}{|G|} \sum_{g \in G} \xi_i^*(g) \rho(g), \ee 
where $\xi_i(g)=\trace( \rho(g))$ is the character of the irreducible representation $V_i$.
Therefore the characteristic function of the projector onto the $i$-th sector reads
\begin{align}
    \chi_{\textbf{a}}(\Pi_i) & = \frac1d \trace (D_{\textbf{a}}^\dagger \Pi_i) \nonumber \\
    & =\frac1d \frac{\dim(V_i)}{|G|}\sum_{g \in G} \xi_i^*(g)\trace (D_{\textbf{a}}^\dagger\rho(g)) \nonumber \\
    & = \frac{\dim(V_i)}{|G|}\sum_{g \in G} \xi_i^*(g) \chi_{\textbf{a}}(\rho(g)).
\end{align}
The invariant subspace $V_0$ is the subspace of $V$ whose vectors remain unchanged under the action of all elements of the group $G$. This subspace corresponds to the sector in the decomposition that transforms according to the trivial representation. For this representation, the character is $\chi_0(g) = 1$ for all $g \in G$, and its dimension is $\dim(V_{0}) = 1$. Substituting these values into the general formula, we obtain the following expression for the projector onto the invariant subspace and its characteristic function,
\begin{align}
   \Pi_0 = \frac{1}{|G|} \sum_{g \in G} \rho(g)  &&  \chi_{\textbf{a}}(\Pi_0) 
    & = \frac{1}{|G|}\sum_{g \in G}\chi_{\textbf{a}}(\rho(g)),
\end{align} 
which shows that the characteristic function of the projector $\chi_{\textbf{a}}(\Pi_0)$ is a linear combination of the characteristic functions of (the representations of) all the group elements.

\section{ASE Gap for random subspaces of fixed dimension}
\label{overall}
The behavior of the average stabilizer entropy gap is highly dependent on the choice of subspace. It is natural to ask, however: what kind of gap should one expect for a subspace chosen purely at random? In pursuing this question, we are led to compute the ASE \emph{averaged over all subspaces of a fixed dimension} $d_S$, evaluating in particular
\begin{align}
    \mathbb{E}_{U, \mathcal{E}}[M_B(\mathcal{E}(U\psi U^\dagger))] = 1-d_B \binom{d_S+3}{4}^{-1}\mathbb{E}_\mathcal{E}\left[ \trace\left(Q_B \Pi_{\mathcal{E}}^{\otimes 4}\Pi_{\text{sym}^4}^B \right)\right],\label{average_sub}
\end{align}
where now the expectation value is over not only all states in a given subspace but also all embedding maps $\mathcal{E}$. We arrive at the following theorem:
\begin{theorem}[ASE gap for random subspaces]
    The ASE gap for a subspace of fixed dimension $d_S$ is on average the difference between the ASE of the larger space $\mathcal{H}_{d_B}$ and the ASE of the small space $\mathcal{H}_{d_S}$,
    \begin{equation}
        \mathbb{E}_\mathcal{E}\big[\triangle M(\mathcal{E})\big] = \mathbb{E}_U[M_B(U\psi U^\dagger)]  - \mathbb{E}_U[M_S(U\psi U^\dagger)].
    \end{equation}
\end{theorem}

\begin{proof}
To prove the result we just need to evaluate \cref{average_sub}. Let $\Pi=\sum_{i=1}^{d_S}|i\rangle\langle i|$ where $\{|i\rangle\}_{i=1}^{d_B}$ is the computational basis on $\mathcal{H}_{d_B}$. We may always write $\Pi_\mathcal{E}=V\Pi V^\dagger$ for some choice of unitary $V\in U(d_B)$, and thus translate our average over subspaces into an average over the unitary group. We observe that $V^{\otimes 4}\Pi^{\otimes 4}V^{\dagger \otimes 4}\Pi_{\text{sym}^4}^B=V^{\otimes 4}(\Pi^{\otimes 4}\Pi_{\text{sym}^4}^B)V^{\dagger \otimes 4}$ since by construction $\Pi_{\text{sym}^4}^B$ commutes with all fourth tensor powers of unitaries. Moreover, it is straightforward to show that $\Pi^{\otimes 4}\Pi_{\text{sym}^4}^B=\mathcal{E}^{\otimes 4}(\Pi_{\text{sym}^4}^S)=\Pi_{\text{sym}^4}^{S\subset B}$ is just the symmetric projector supported on the subspace. Using the Weingarten calculus (see \cref{App_Weingarten} for a short review), we may evaluate the average itself as
\ba 
\exv_{V}\left[ V^{\ot 4} \Pi_{\text{sym}^4}^{S\subset B} V^{\dag \ot 4}\right]  = & \sum_{\sigma,\tau\in S_4} \wg (\sigma,\tau)\trace\left(\Pi_{\text{sym}^4}^{S\subset B} T_{\sigma}\right) T_{\tau},
\ea 
so that
\ba 
\mathbb{E}_\mathcal{E}\left[\trace\left(Q_B\Pi_{\mathcal{E}}^{\otimes 4}\Pi_{\text{sym}^4}^B \right)\right] &=&
\sum_{\sigma,\tau\in S_4}\wg(\sigma ,\tau)\trace\left(\Pi_{\text{sym}^4}^{S\subset B} T_{\sigma}\right) \trace\left(Q_B T_{\tau}\right) .
\ea 
Using the fact that $\Pi_{\text{sym}^4}^{S\subset B}=\frac{1}{4!}\sum_{\pi\in S_4} T_{\pi}^{S\subset B}$, and
\begin{equation}
    T_{\alpha}^{S\subset B}T_{\beta}^{B}=T_{\alpha}^{S\subset B}T_{\beta}^{S\subset B}=T_{\alpha \beta}^{S \subset B},
\end{equation} 
along with the following property of the Weingarten coefficients proved in \cref{App_Weingarten},
\be 
\sum_{\sigma} \text{Wg}(\sigma,\tau) = \frac{1}{4!}\binom{d_B+3}{4}^{-1} ,
\ee
we arrive at
\ba 
\mathbb{E}_\mathcal{E}\left[\trace\left(Q_B\Pi_{\mathcal{E}}^{\otimes 4}\Pi_{\text{sym}^4}^B \right)\right] &=&\sum_{\sigma,\tau\in S_4}\wg(\sigma,\tau)\trace\left(Q_B T_{\tau}\right)\frac{1}{4!}\sum_{\pi\in S_4} \trace\left(T_{\pi\sigma}^{S\subset B}\right) \\
&=& \binom{d_S+3}{4}\sum_{\sigma,\tau}\wg(\sigma,\tau)\trace\left(Q_B T_{\tau}\right)\\
&=&\binom{d_S+3}{4} \binom{d_B+3}{4}^{-1}\trace\left(Q_B \Pi_{\text{sym}^4}^B\right). \label{fin_avg}
\ea
We conclude that
\begin{align}
    \mathbb{E}_{U, \mathcal{E}}[M_B(\mathcal{E}(U\psi U^\dagger))]  &= 1-d_B \binom{d_B+3}{4}^{-1}\trace\left(Q_B \Pi_{\text{sym}^4}^B\right),
\end{align}
which is, in fact, just the stabilizer entropy averaged over the big space. From the definition of the ASE gap, we then have our claim.
\end{proof}

\begin{figure}[t]
\centering
\includegraphics[width=0.5\textwidth]{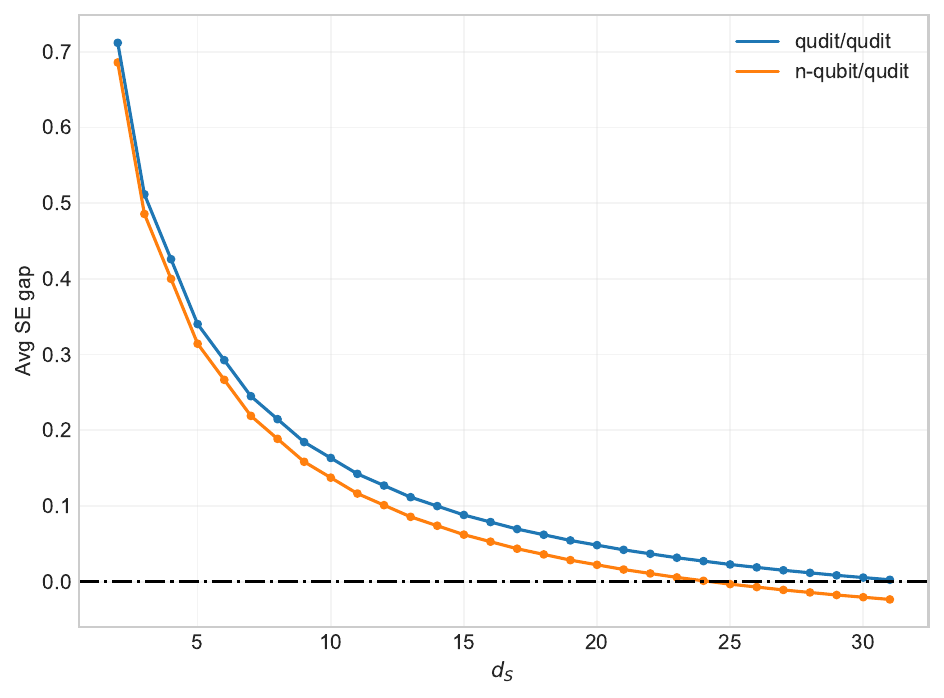}\includegraphics[width=0.5\textwidth]{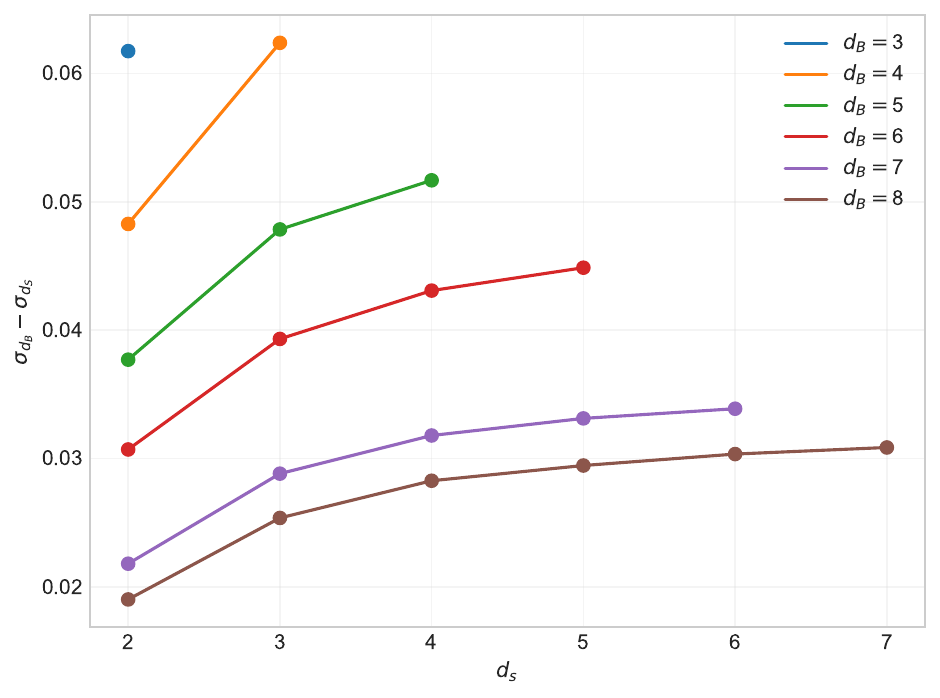}
\caption{On the left: the ASE gap, averaged over subspaces, plotted for $d_B=32$ and for $d_S$ from $2$ to $d_B-1$. In blue, we treat the $\mathcal{H}_{d_B}$ as a single qudit, and in orange we treat it as $n=5$ qubits. On the right: for different qudit dimensions $d_B$, the difference between the standard deviation $\sigma_{d_B}$ of the ASE of $\mathcal{H}_{d_B}$ and the standard deviation $\sigma_{d_S}$ of the ASE averaged over $d_S$-dimensional subspaces. The former was calculated by averaging over 750 random states in $\mathcal{H}_{d_B}$, the latter by averaging over 750 random subspaces.}
\label{fig:double_average}
\end{figure}

We have thus shown that the ASE of a subspace, averaged over all choices of subspace with fixed dimension $d_S$, amounts to the ASE of the entire $d_B$ dimensional space. In other words, for a subspace selected at random, the ASE gap will be on average
\begin{align}
    \mathbb{E}_\mathcal{E}\big[\triangle M(\mathcal{E})\big] &= \mathbb{E}_U[M_B(U\psi U^\dagger)]  - \mathbb{E}_U[M_S(U\psi U^\dagger)] .
\end{align}
The left hand side of Figure \ref{fig:double_average} plots the average gap for $d_B=32$ and for all possible small dimensions $d_S$. We consider $\mathcal{H}_{d_B}$ both as a single qudit and as 5 qubits. Strikingly, when the larger space is treated as $n$-qubits, the ASE gap will generically go \emph{negative} as $d_S$ approaches $d_B$, implying  fewer resources are required to simulate the embedded system than are required to simulate the system treated on its own. Finally, we observe that the variance of the ASE of $\mathcal{H}_{d_B}$ will differ if one averages over states as compared to subspaces. On the right hand side of Figure \ref{fig:double_average}, we plot the difference between the standard deviation $\sigma_{d_B}$ of the ASE averaged over all states in $\mathcal{H}_{d_B}$ and the standard deviation $\sigma_{d_S}$ of the ASE averaged over subspaces of fixed dimension $d_S$. The  $\sigma_{d_S}$ is lower, and goes to 0 as $d_S\rightarrow d_B$.

\section{Subspaces with extremal nonstabilizerness}
\label{extreme}

\begin{figure}[t]
\centering
\includegraphics[width=1\textwidth]{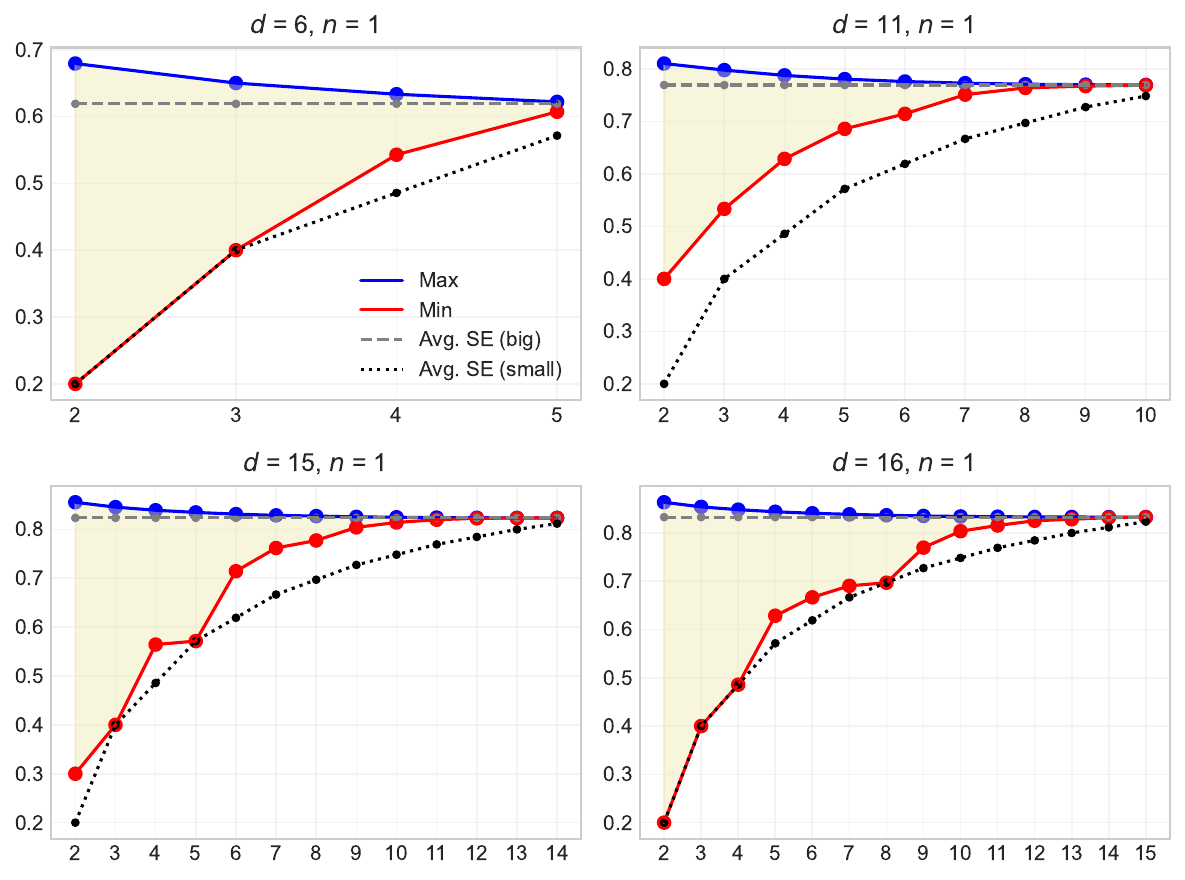}
\caption{For a variety of big Hilbert space dimensions $d_B$, we plot for each small dimension $d_S$ from $2$ to $d_B-1$ the minimal average SE achievable over all choices of subspace as well as the maximum. The dashed horizontal line depicts the average SE of the big space, while the dotted line depicts the average SE of the small space. When the dotted line intersects the red, zero magic gap is achieved.}
\label{fig:qudit_qudit_select}
\end{figure}

\subsection{Numerical results}

If one is free to choose a subspace in which to encode a quantum system, which subspace should one choose ? Fixing a set of WH operators $\{D_{\B{b}}\}$ on the ``big'' Hilbert space $\mathcal{H}_{d_B}=\mathcal{H}_d^{\otimes n}$ as well as a dimension $d_S$ of the desired subspace, one potential answer is to choose a subspace which achieves the least average stabilizer entropy. Then on average, the preparation of any state of the encoded subsystem will require the least magic state injections in the quantum simulation case, and the least samples in the classical case.

\begin{figure}[t]
\centering
\includegraphics[width=1\textwidth]{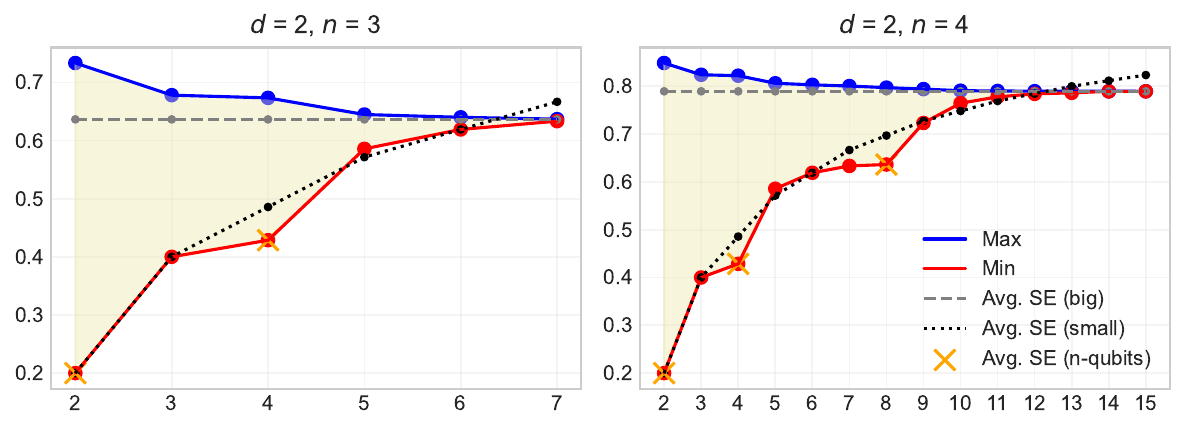}\\
\includegraphics[width=0.495\textwidth]{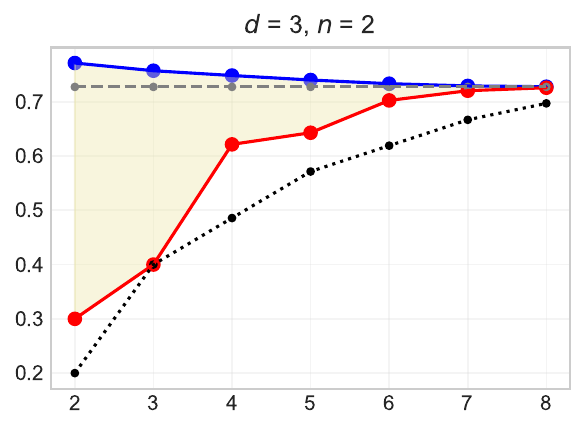}
\includegraphics[width=0.495\textwidth]{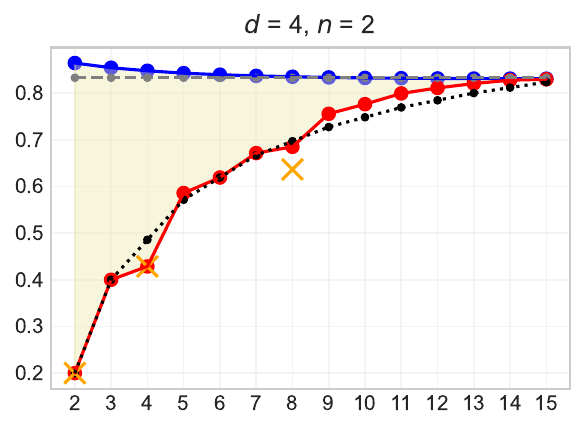}
\caption{The same in several multiqudit cases where the larger Hilbert space is $\mathcal{H}_d^{\otimes n}$. The orange X's depict the ASE of the small space considered as $n$-qubits when the subspace dimension $d_S$ is a power of 2. }
\label{fig:multiqudits}
\end{figure}

From a theoretical point of view, it is of interest to consider extremizing the ASE in both directions, considering for a fixed WH structure on the larger Hilbert space and fixed dimension $d_S$ of a subspace, those subspaces which minimize and maximize the average stabilizer entropy. Figure \ref{fig:qudit_qudit_select} displays the results of numerically extremizing the ASE over choices of subspace for a variety of dimensions with respect to the single qudit WH group. (See \cref{qubit_extrema} for the results for all dimension $d_B=3$ to $16$.) The BFGS algorithm \cite{Luenberger2008} was used to perform the calculations, and each optimization was run $\ge 100$ times. For lower dimensions, the ASE in the cost function was calculated exactly using Eq.\! \ref{bigguy} for half the runs. The rest of the time the ASE was approximated using a Monte Carlo scheme.

It is clear from the depicted results that as $d_S$ approaches $d_B$, the extremal ASE values approach the average SE of the full Hilbert space: this is consistent with the results of Section \ref{overall}. It is also clear that it is possible to achieve zero average magic gap for certain combinations of $d_S$ and $d_B$, in particular when $d_S$ is a factor of $d_B$. But this is not the only case: evidently, when $d_B$ is a power of 2, a zero ASE subspace of dimension three is always achievable. The asymmetry of these graphs is also notable: the maximal achievable ASE is already close to the ASE of the larger Hilbert space.

The top row in Figure \ref{fig:multiqudits} shows the results of the same exercise where we consider the $n$-qubit WH group on the larger space for $n=3,4$. We have found that for $n=2$, as in the $d_B=4$ qudit case, zero ASE gap is achievable for $d_S=2,3$. In fact, in \cref{fig:sym_qubits}, we shall see an example of the latter. More interestingly, for $n=3$, zero ASE gap is achievable for $d_S=2,3,6$ and \emph{negative} SE gap for $d_S=4,(7)$, where the parentheses denote the fact that in this case the average over the small space is already larger than the average over the big space, as noted in \cref{overall}. In fact, we can see that the gap is negative for $d_S=4$ when the subspace is treated as a qudit, precisely because zero gap is attainable when the subspace is treated as two qubits. Similarly, for $n=4$, zero SE gap is apparently attainable for $d_S=2,3,6$, while negative SE gap is possible for $d_S=4,7,8,9,12, (13),(14),(15)$. The cases of $d_S=4,8$, in fact, correspond to zero magic gap when the subspaces are treated as two and three qubits respectively. The lower two plots in \cref{fig:multiqudits} depict the results for two $d=3$ systems and two $d=4$ systems. The former is comparable to the single qudit case: zero magic gap is achievable only for $d_S=3$ which divides 9. On the other hand, for two ququarts, zero magic gap is achievable for $d_S=2,3,6$, while negative magic gap is achievable for $d_S=4,8$. While $d_S=4$ corresponds to zero magic gap when the subspace is treated as two qubits, it is apparently not the case that when $d_S=8$, one can achieve the ASE of three qubits. We note that if one could, then the minimal ASE would not be a monotonically increasing function of dimension.

\subsection{\texorpdfstring{A sufficient condition for ASE gap $\leq 0 $ }{A sufficient condition for ASE gap <= 0}}\label{zerogapsec}

In the previous section, we concluded from our numerics that in the single qudit case, if $d_S$ is a factor of $d_B$, one may find subspaces with zero average stabilizer entropy gap, and that for big Hilbert spaces of even dimension equipped with a multiqudit WH group, negative ASE gap is also achievable. We now establish a sufficient condition for a subspace to have ASE gap $\leq 0$ in the case the subspace corresponds to a particular type of stabilizer code. 

Let $d_B=d^n$ where $d$ is the local dimension, and $n$ is the number of qudits, equipped with a set of WH operators $\{D_{\B{a}}\}$. In \cref{stab_subspaces}, we showed that a stabilizer codespace corresponds to an abelian subgroup of the WH group, which can be presented as a totally isotropic set $\mathcal{S}$ of symplectic indices which are closed under addition mod $d$, which have pairwise vanishing symplectic product, and which contain the 0 vector, as well as a choice of group homomorphism $f(\B{a})$. The projector onto the codespace is then given by
\begin{align}
\Pi_{\mathcal{E}_{\mathcal{S},f}} = \frac{1}{|\mathcal{S}|}\sum_{\B{a}\in \mathcal{S}} \omega^{f(\B{a})}D_{\B{a}}.
\end{align}
We recall that there exists a totally isotropic subgroup of size $|\mathcal{S}|$ iff $|\mathcal{S}|$ divides $d_B$, and the dimension of the codespace is $d_S=d_B/|\mathcal{S}|$. In \cref{proofzerogap}, we prove the following theorem:
\begin{theorem}[ASE gap for stabilizer codespaces]\label{zerogap}
    The average stabilizer gap of a stabilizer codespace with  isotropic set $\mathcal{S}$ and trivial group homomorphism $f(\B{a})=0\,\,\,\forall \B{a}\in \mathcal{S}$, with corresponding isometry $\mathcal{E}_\mathcal{S}$, is 
   \begin{align}
	\Delta M(\mathcal{E}_\mathcal{S})&= \frac{\alpha d_B - |A_\mathcal{S}|d_S}{d_B(d_S+1)(d_S+2)(d_S+3)} && \alpha = \begin{cases}1, & d_S \text{ odd (multiqudit)} \\
	4, & d_S \text{ even (multiqudit)}\\
    d_S^2, &d_S=2^m \text{ (multiqubit)}, \end{cases}
\end{align}
where $|A_\mathcal{S}|$ is the cardinality of the set
\begin{align}
A_\mathcal{S} = \Big\{ \B{a} \in \mathbb{Z}_d^{2n}\ | \ 2\B{a}\in \mathcal{S}\quad and\quad 	\B{a}\in S^\perp \Big\}.
\end{align}
with $S^\perp\equiv\{ \B{a}\in \mathbb{Z}_d^{2n} \,\,| \,\,\forall \B{b}\in \mathcal{S},\,[\B{a}, \B{b}]=0\}$
\end{theorem}
We now show that $\Delta M(\mathcal{E}_\mathcal{S})$ is easily calculable in at least three general cases. First, suppose that the larger Hilbert space is treated as $n$-qubits. Since $d=2$, we have that $\forall \B{a}$, $ 2\B{a}=0$. Thus the first condition on $A_\mathcal{S}$ is satisfied trivially, and so $A_\mathcal{S}=\mathcal{S}^\perp$. Meanwhile, Theorem 25 of \cite{gross} tells us that in general
\begin{align}
    |S^\perp| = \frac{|\mathbb{Z}_d^{2n}|}{|S|}= \frac{d^{2n}}{d^n/d_S}=d^nd_S=d_B d_S.
\end{align}
Our \cref{zerogap} then implies that if the subspace is treated as $m$-qubits, $\Delta M(\mathcal{E}_\mathcal{S})=0$, and otherwise, the average magic gap is negative. In fact, we shall see an example of this in \cref{422}. 

On the other hand, suppose that the code space is treated as a qudit, or multiple qudits. First, we consider the case that $d_B$ is odd. Since $|\mathcal{S}|$ is an integer, $d_S$ must be odd as well. The first condition on the set $A_\mathcal{S}$ is that $
\forall \,\B{a} \,\,\exists \,\B{c}\in \mathcal{S}$ such that  $2\B{a} \,\equiv\, \B{c}$.
Since $d$ is odd, $2$ has a unique inverse: $\frac{1}{2} = \frac{d+1}{2} \in \mathbb{Z}_{d}$, so that $\B{a}=\frac{d+1}{2}\B{c}$ uniquely, from which we conclude that $\B{a}\in \mathcal{S}$. By the isotropy of $\mathcal{S}$, the second condition is automatically satisfied, and so $A_\mathcal{S}=\mathcal{S}$ itself. By \cref{zerogap}, $\Delta M(\mathcal{E}_\mathcal{S})=0$. 

Finally, suppose that a) $d_B$ is even, b) $d_S$ is even and c) $\mathcal{S}\subseteq 2\mathbb{Z}_d^{2n}$, that is, the elements of $\mathcal{S}$ are all even. Since $d$ is even, 2 has no unique inverse. In general, a linear congruence
\begin{align}
ka \equiv c \pmod d
\end{align}
has solutions if and only if $\text{gcd}(k,d) \ |\  c$ and then  $\text{gcd}(k,d)$ is the number of solutions. Here $k=2$, so that $\text{gcd}(2,d)=2$. Thus for the congruence to have a solution, each component of $\B{c}$ must be divisible by 2, and this is indeed the case since by assumption $S\subseteq 2\mathbb{Z}_d^{2n}$. There are two solutions for each component, namely those that satisfy the congruence
\begin{align}
a_i \equiv\frac{1}{2}c_i \pmod{d/2},
\end{align}
that is,
\begin{align}
a_i = 	\frac{1}{2}c_i \mod d && a_i=\frac{1}{2}c_i + \frac{d}{2} \mod d.
\end{align}
Therefore for a given $\B{c}\in \mathcal{S}$, there are $2^{2n}=4^n$ solutions each corresponding to a valid $\B{a}$. Moreover, since $2\B{a}\in \mathcal{S}$,
\begin{align}
\forall \B{c}\in \mathcal{S}: [2\B{a}, \B{c}]=[\B{a}, 2\B{c}]=0.	
\end{align}
But since $\mathcal{S}\subseteq 2\mathbb{Z}_d^{2n}$, \emph{any} element of $\mathcal{S}$ can be written as $2\B{c}$ for some $\B{c}\in \mathcal{S}$, and so $\B{a}\in \mathcal{S}^\perp$ and the second condition of the set $A_\mathcal{S}$ is satisfied. Thus $|A_\mathcal{S}|=4^n |\mathcal{S}|$, and
\begin{align}
	\Delta M(\mathcal{E}_\mathcal{S})&= \frac{4 - 4^n}{(d_S+1)(d_S+2)(d_S+3)}.
\end{align}
By \cref{zerogap}, if $n=1$, then $\Delta M(\mathcal{E}_\mathcal{S})=0$, and more remarkably, if $n>1$, the average stabilizer entropy gap is negative.

In closing, we emphasize that we have given sufficient conditions for $\Delta M(\mathcal{E})\leq 0$ in the special case that the subspace corresponds to a stabilizer code (with trivial group homomorphism). At the same time, our numerical analysis has shown that there exist zero ASE gap subspaces whose corresponding projectors are much more complicated linear combinations of WH operators, not corresponding to stabilizer codes. Moreover, as observed in Section \ref{extreme}, it appears that if the overall system is treated as a qudit and $d_B$ is a power of 2, then there exist zero ASE gap subspaces for $d_S=3$. Not only that, but when the overall system is treated as $n$-qubits, one can achieve zero SE gap when $d_S$ has a factor of 3, and similarly in the case of two ququarts. As we will see in Section \ref{spin}, when $d=2, n=2$, the symmetric subspace provides an example of a three dimensional zero ASE gap subspace. We hope to clarify these observations in future work.

\section{Support on the Complement}
\label{complement}
\begin{figure}[t]
\centering
\includegraphics[width=1\textwidth]{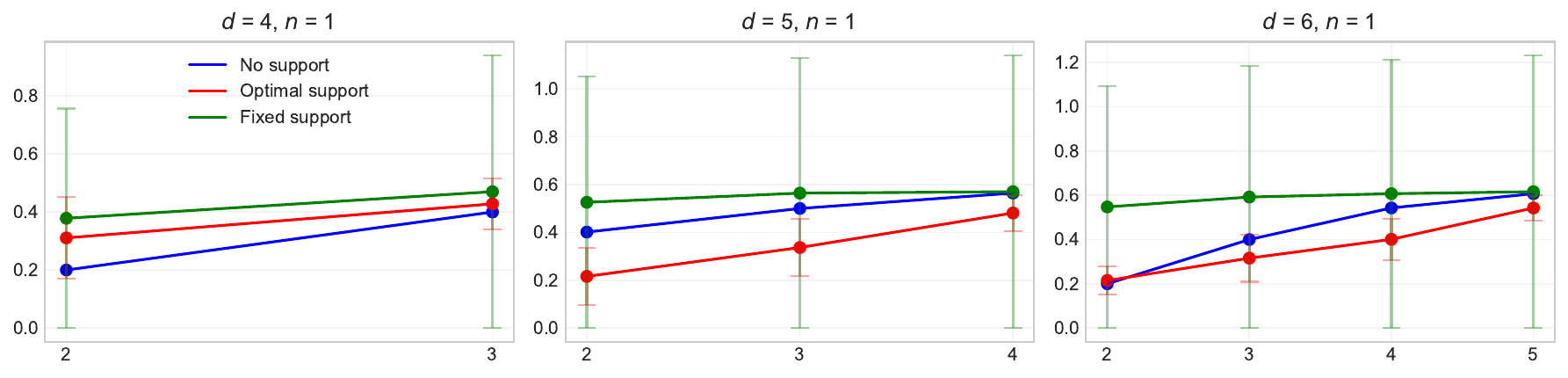}\\
\includegraphics[width=0.35\textwidth]{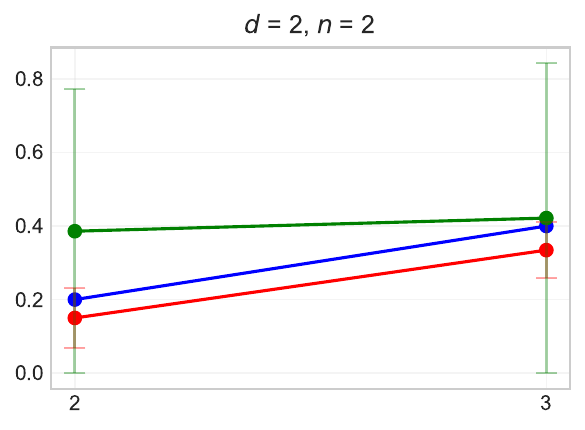}
\includegraphics[width=0.35\textwidth]{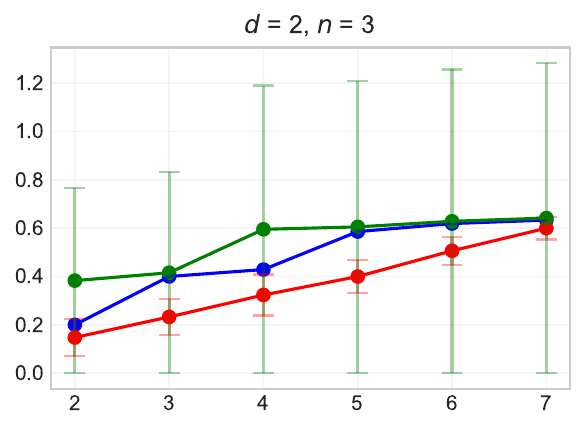}
\caption{For different values of $n$ and $d$, we plot in blue the minimal achievable ASE of a subspace of dimension $d_S$. In red, we plot the ASE obtained by adding in the optimal complement (in even superposition) for each of 750 random states. In green, we plot the ASE obtained by adding in the optimal fixed complement (in even superposition), again for each of 750 random states. In the latter two cases, the standard deviations are plotted as error bars.}
\label{fig:complements}
\end{figure}

So far we have considered embedding a $d_S$ dimensional state in a $d_B$ dimensional Hilbert space by means of an isometry $\mathcal{E}: \mathcal{H}_{d_S}\rightarrow \mathcal{H}_{d_B}$ so that the state $\mathcal{E}(\psi)$ is supported only on the chosen subspace. Since a priori, the WH operators $\{D_\B{a}\}$ on $\mathcal{H}_{d_B}$ do not respect the choice of subspace, it is possible that lower stabilizer entropy could be achieved by a state of the form
\begin{align}
\label{complementary_form}
    |\phi\rangle = \alpha \mathcal{E}(|\psi\rangle) + \beta |\kappa\rangle,
\end{align}
where $|\kappa\rangle$ is supported only on the complement of the subspace. Any state which can be written in the form of $|\phi\rangle$ will project to $|\psi\rangle$ on the subspace, and will be equivalent as far as operations confined to the subspace are concerned, but with respect to the WH operators on $\mathcal{H}_{d_B}$, we may have $M(|\phi\rangle) < M(\mathcal{E}(|\psi\rangle)$. It is then natural to consider minimizing the SE of $|\phi\rangle$ over the choice of support on the complement. In fact, in what follows, we will set $\alpha=\beta=\frac{1}{\sqrt{2}}$ in Eq.\! \ref{complementary_form}. Without such a restriction, it is possible that the support on the subspace itself may be driven to 0. Recall that any set of measurement operators $\{E_i\}$ must sum to the identity on the full Hilbert space, that is, $\sum_i E_i=I$. Therefore to perform a measurement on the subspace, one must include at least a projector onto the complement as a measurement operator, corresponding to ``no answer.'' If the support of a state on the subspace has very small norm, then with overwhelming probability such a measurement will yield ``no answer.'' In practical applications, it may be tolerable to have a small norm on the subspace, if the resource advantage in terms of magic states outweighs the increased number of samples necessary to obtain expectation values. For simplicity, we confine ourselves to this particular choice of norm. 

We may now contemplate the ASE of a subspace where for each state we consider the optimal support on the complement (which minimizes the SE). It is difficult, however, to interpret this in operational terms.  Suppose one wants to apply a unitary on the subspace. To keep the SE minimized, one will have to apply a unitary on the complement which takes the optimal support of the initial state to the optimal support of the final state. Unless the computation has a very special symmetry, one's unitaries must then all be state-dependent. On the other hand, one could ask a different question: given a choice of subspace, which is the \emph{single} optimal support on the complement, which minimizes the average SE? This is quite natural: if one can prepare this optimal support, then it will be left invariant under any operations supported solely on the subspace.

In \cref{fig:complements}, we consider for different $n$, $d$, and $d_S$, a subspace which achieves the minimal ASE. We compare the value of the ASE a) without support on the complement (blue), b) with optimal support for each state considered in the average (red), c) with the optimal fixed support in the complement (green): in the latter two cases, the standard deviation is plotted in the form of error bars. In each case, we see that adding a fixed support in the complement does worse than no support at all, and in fact leads to wide variation in the SE. On the other hand, generally speaking one may achieve lower ASE by adding in optimal support for each individual state, but as we have argued, this is unlikely to be relevant in practice. It is notable, however, that this does not appear true for $d=4$.  On the other hand, \cref{fig:complements_rel_change} shows the relative change in ASE averaged over 100 randomly chosen subspaces, using in each case the optimal fixed complement for that subspace, for different choices of $d, n,$ and $d_S$. From this it is clear, that for a typical subspace, it is ideal to have support in the complement, particularly when $d_S$ is much smaller than $d_B$.

\begin{figure}[t]
\centering
\includegraphics[width=1\textwidth]{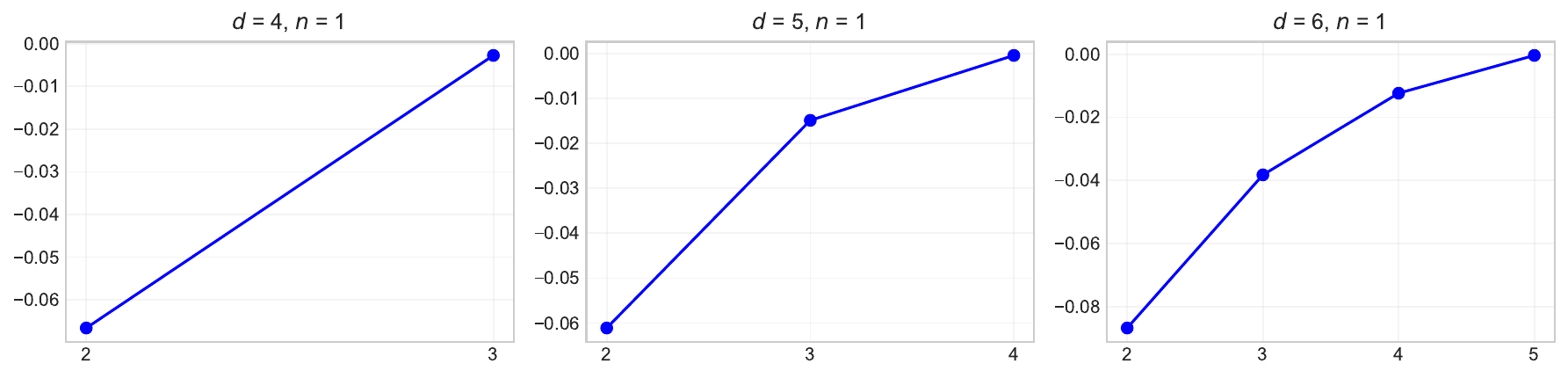}\\
\includegraphics[width=0.35\textwidth]{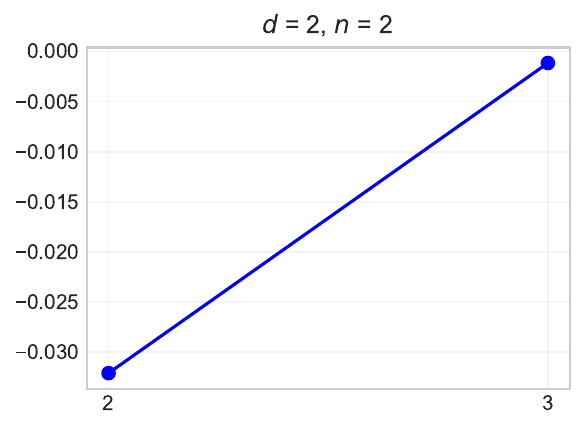}
\includegraphics[width=0.35\textwidth]{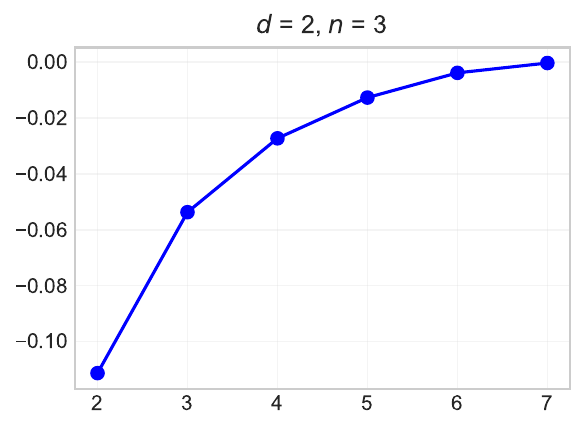}
\caption{For different values of $n$ and $d$, we plot the average relative change in ASE over 100 random choices of subspaces with dimension $d_S$, where the ASE is calculated by adding in the optimal fixed complement (in even superposition, for 750 random states).}
\label{fig:complements_rel_change}
\end{figure}

\section{Examples}
\label{examples}

\subsection{\texorpdfstring{The ground space of frustration-free Hamiltonian}{The ground space of frustration-free Hamiltonian}}

An instructive example of a physically relevant subspace is the ground state space (GSS) of a frustration-free Hamiltonian \cite{doi:10.1137/08072689X,PhysRevA.84.042338}. By construction, a frustration-free Hamiltonian is a sum of local Hamiltonians $H = \sum_i h_i$ that admits ground states which minimize each local term $h_i$ simultaneously. Thus the ground space can be expressed as the intersection of the local kernels $\mathcal{G} = \bigcap_i \ker(h_i)$. In many-body physics, the ground space plays a central role since it encodes the system’s long-range entanglement patterns and provides a natural setting for robust information encoding.

The authors of \cite{chen_ground-state_2012} have proposed a geometric framework in which to study such ground state spaces, and as an illustration construct explicitly the GSS for a two-local frustration-free Hamiltonian, that is, a Hamiltonian composed of interaction terms acting non-trivially on at most two particles (Example 6).  Considering a three-qubit system, the GSS is two-dimensional,
\be
S=\text{span}\{\frac{1}{\sqrt{3}}(\ket{001}+\ket{010}+\ket{100}), \ket{000}\}.
\ee
On the one hand, the intrinsic ASE for a two-dimensional space is $\frac{1}{5}$. On the other hand, considered as a two-dimensional subspace of $\mathcal{H}_{2}^{\otimes 3}$, the extrinsic ASE is $\frac{5}{9}$, giving an ASE gap of $\frac{16}{45}= 0.3\overline{5}$. Considered as a two-dimensional subspace of $\mathcal{H}_8$, however, the extrinsic ASE is $\frac{83}{135}=0.6\overline{148}$, giving an ASE gap of $\frac{56}{135}=0.4\overline{148}$.

\subsection{\texorpdfstring{Symmetric subspace of $n$-qubits}{Symmetric subspace of n-qubits}}

\label{spin}
\begin{figure}[t]
\label{fig:sym_qubits}
\centering
\includegraphics[width=0.5\textwidth]{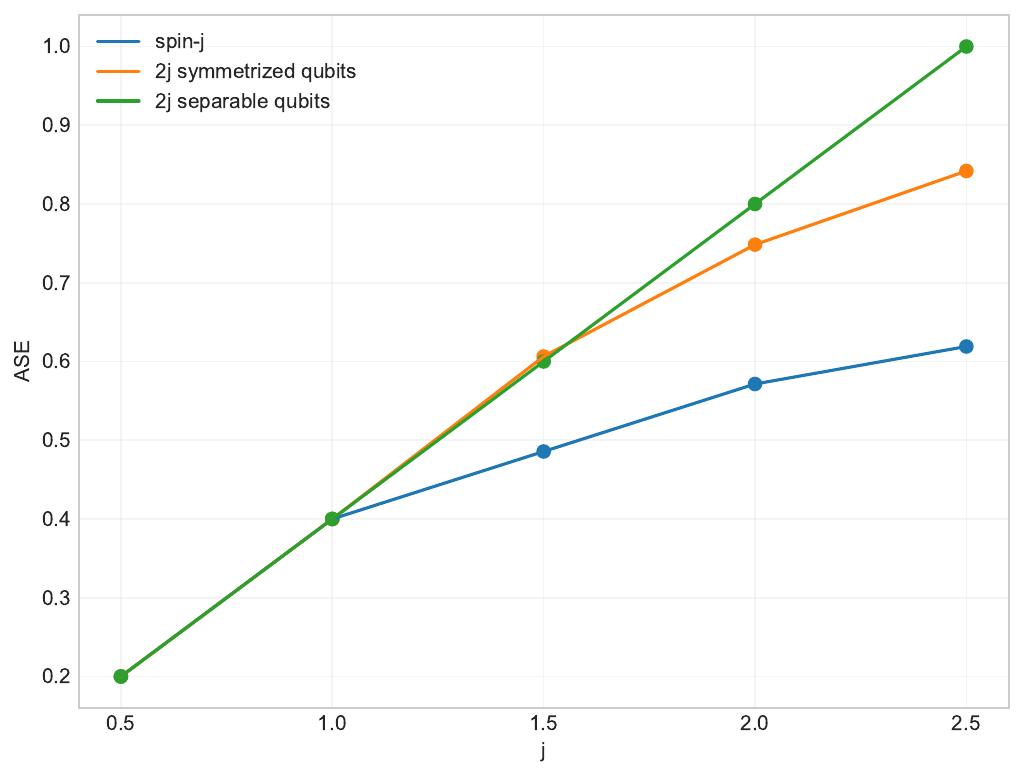}
\caption{The ASE of three representations of a spin-$j$ state: a) intrinsically, in a $2j+1$ dimensional Hilbert space, b) extrinsically, as $2j$ symmetrized qubits, and c) as $2j$ separable qubits using the Majorana stellar representation.}
\end{figure}

The dimension of the symmetric subspace of $\mathcal{H}_d^{\otimes n}$ is $\binom{d+n-1}{n}$. In particular, when $d=2$, the dimension of the symmetric subspace becomes $n+1$. Thus we may encode a spin-$j$ state, which is $2j+1$ dimensional, into the symmetric subspace of $n=2j$ qubits. For example, for $j=1$,
\begin{align*}
|1,1\rangle &\rightarrow |\!\uparrow \uparrow \rangle  \\
|1, 0\rangle &\rightarrow  \frac{1}{\sqrt{2}}\Big(|\!\uparrow \downarrow \rangle + |\!\downarrow \uparrow \rangle \Big)\\
|1, -1\rangle &\rightarrow |\!\downarrow \downarrow \rangle . 
\end{align*}
In general, the spin-$j$ basis vector $|j, m\rangle$ is taken to the $2j$-qubit state which is an even superposition of all states with $j+m$ qubits $\uparrow$ and $j-m$ qubits $\downarrow$. We may then compare the intrinsic ASE of the $d_S=2j+1$ states calculated with respect to the qudit WH group against the extrinsic ASE of the $d_B = 2^{2j}$ states calculated with respect to the $n$-qubit WH group. At the same time, we may also express a spin-$j$ state as a direct sum of $2j$ (separable) qubits. This construction makes use of the \emph{Majorana stellar representation} \cite{Bengtsson_Życzkowski_2017, Liu_2017}. Given a pure spin-$j$ state $|\psi\rangle$, we may form a polynomial of a single complex variable
\begin{align}
p(z) &= \sum_{m=-j}^{j} (-1)^{j-m} \sqrt{\binom{2j}{j-m}} \langle j, m|\psi\rangle z^{j+m},
\end{align}
and decompose it into its roots $\{\alpha_i\}$. If there are less than $2j$ roots, we add $\infty$ as a root enough times for $|\{\alpha_i\}|$ to be $2j$. We may then reverse stereographically project these roots from $\mathbb{C}\cup \infty$ to the Riemann sphere via
\begin{align}
\alpha \rightarrow \begin{cases}(0,0,-1), & \alpha=\infty\\
\left(2\Re[\alpha], 2\Im[\alpha], \frac{1-|\alpha|^2}{1+|\alpha|^2}\right), & \text{otherwise}\end{cases}.	
\end{align}
In this way, up to phase, each pure spin-$j$ state can be mapped bijectively to a constellation of $2j$ points (or stars) on the sphere. On the one hand, this is a natural generalization of the Bloch sphere representation of a qubit, in particular as $\text{SU}(2)$ rotations of the spin-$j$ state correspond to rotations of the full constellation; on the other hand, these special points can be related operationally to the zeros of the spin-coherent wavefunction of the state. Alternatively, one may map the roots to qubit states via 
\begin{align}
\alpha \rightarrow \begin{cases}(0,1)^T, & \alpha=\infty\\
\frac{1}{\sqrt{1+|\alpha|^2}}(1, \alpha)^T, & \text{otherwise}.\end{cases}.
\end{align}
If one takes the permutation symmetric tensor product of these $2j$ qubit states, one recovers up to phase precisely the permutation symmetric state originally constructed.

The results of an exact computation are shown in Fig. \ref{fig:sym_qubits}. Notably, for $j=1$, the ASE gap is zero in both cases: the average intrinsic stabilizer entropy is equal to the average extrinsic stabilizer entropy calculated with respect to both the symmetrized and separable qubits. For all other $j$, the ASE gap is positive. Notably, only for $j=\frac{3}{2}$ does the average SE of the symmetrized qubits exceed that of the separable qubits.

\subsection{Quantum polyhedra}

Given the tensor product of $n$ spin Hilbert spaces $\otimes_{i=1}^n\mathcal{H}_{j_i}$, we may consider the Clebsch-Gordan decomposition into a direct sum of spin sectors labeled by total spin $J$, each of which may appear with multiplicity $\mu_J$,
\begin{align}
\bigotimes_{i} \mathcal{H}_{j_i} =	\bigoplus_J\mathcal{H}_{J} \otimes \mathbb{C}^{\mu_J} .
\end{align}
For example, we may decompose four spin-$\frac{1}{2}$ spins into
\begin{align}
\mathcal{H}_{\frac{1}{2}}^{\otimes 4}	&= 2\mathcal{H}_0 \oplus 3 \mathcal{H}_1 \oplus \mathcal{H}_2,
\end{align}
where the coefficients signify the multiplicity. In this case, the spin-$0$ sector has multiplicity two, and so can be regarded as a logical qubit invariant under collective SU$(2)$ rotations of the four physical qubits. For example, in the framework of a background independent theory of quantum gravity \cite{rovelli2015covariant,Oriti_2016}, the restriction to the spin-$0$ subspace is precisely the implementation of the Gauss law constraint in the theory, and geometrically, it is a quantization of the Minkowski condition for the closure of a polyhedron. The four physical qubits represent edges in a spin-network or equivalently, facets of the polyhedron.

We consider quantum tetrahedra with spin-$\frac{1}{2}$ facets and spin-$1$ facets, as well as a quantum cube with six spin-$\frac{1}{2}$ facets. The extrinsic ASE for the spin-$\frac{1}{2}$ tetrahedron was already calculated exactly in \cite{PhysRevD.109.126008}, yielding $\frac{17}{45}$.  For the spin-$1$ tetrahedron, the Clebsch-Gordan decomposition gives
\begin{align}
    \mathcal{H}_1^{\otimes 4} = 3 \mathcal{H}_0 \oplus 6 \mathcal{H}_6 \oplus 6 \mathcal{H}_2 \oplus 3\mathcal{H}_3 \oplus \mathcal{H}_4,
\end{align}
so that $d_S=3$. For the spin-$\frac{1}{2}$ cube, we have
\begin{align}
    \mathcal{H}_{\frac{1}{2}}^{\otimes 6} = 5 \mathcal{H}_0 \oplus 9 \mathcal{H}_1 \oplus 5 \mathcal{H}_2 \oplus \mathcal{H}_3,
\end{align}
so that $d_S=5$. To approximate the ASE in the latter two cases, we consider the average of 20 Monte Carlo runs each with 1000 samples. The results are tabulated below, demonstrating in each case a positive ASE gap.
\begin{center}
\begin{tabular}{||c | c | c ||} 
 \hline
Polyhedron & Extrinsic ASE & Intrinsic SE \\
 \hline\hline4 spin-$\frac{1}{2}$ & $0.3\overline{7}$ &0.2\\
\hline
4 spin-$1$ & $0.85183 \pm 0.00119$ &  0.4\\
\hline
6 spin-$\frac{1}{2}$ & $0.75041 \pm 0.002285$ & 0.57142\\
\hline
\end{tabular}
\end{center}

\subsection{\texorpdfstring{$\mathbb{Z}_d$}{Zd} gauge theory}\label{exgauge}
As in the previous example, consider $n$ non-gauge invariant links with Hilbert space $\mathcal{H}_d^{\otimes n}$ on a lattice. For a group $\mathcal{G}$, we can construct a $\mathcal{G}$-gauge invariant site by acting with a gauge invariant projector on the Hilbert spaces of the links. The form of the projector will depend on the group $\mathcal{G}$. In \cite{PhysRevD.109.126008} and \cite{doi:10.1142/S0219887825500033}, the authors explored the ASE gap when the group $\mathcal{G}$ is respectively $SU(2)$ and $\mathbb{Z}_d$. We revisit the results of the $\mathbb{Z}_d$ case, refining them with the aid of the tools we have developed in this paper.

Consider a single $n$-valent site and the $\mathbb{Z}_d$-gauge invariant projector
\begin{equation}
    \Pi_{\mathcal{E}_\mathcal{S}}=\frac{1}{d}\sum_{i=0}^{d-1}(X^i)^{\otimes n}\equiv\frac{1}{|\mathcal{S}|}\sum_{\B{a}\in \mathcal{S}}D_a,
\end{equation}
where $\mathcal{S}=\{(x,0,\dots,x,0),\,\,x\in \mathbb{Z}_d\,\}$. Thus $d_B=d^n$ while $d_S=d_B/|\mathcal{S}|=d^{n-1}$. We can then apply \cref{zerogap} to find that
\begin{align}
	\Delta M(\mathcal{E}_\mathcal{S})&= \frac{\alpha d^n - |A_S|d^{n-1}}{d^n(d^{n-1}+1)(d^{n-1}+2)(d^{n-1}+3)} && \alpha = \begin{cases}1, & d \text{ odd (multiqudit)} \\
	4, & d \text{ even (multiqudit)}\\
    d^{2(n-1)}, &d=2^{n-1} \text{ (multiqubit)}. \end{cases}
\end{align}
As we noted in \cref{zerogapsec}, the gap will be zero whenever $d$ is odd or we are dealing with qubits. Otherwise, we should look at the cardinality of the set $A_\mathcal{S}\equiv\{ \B{a} \in \mathbb{Z}_d^{2n}\ | \ 2\B{a}\in \mathcal{S},	\B{a}\in \mathcal{S}^\perp\}$. Given the simple structure of $\mathcal{S}$, it is possible to explicitly count how many elements are contained in $A_\mathcal{S}$, the details of which can be found in \cref{gauge}. We are led to the following expression of the ASE gap for even local dimensions,
\begin{equation}
\label{even_local}
    \Delta M(\mathcal{E}_\mathcal{S})=\frac{4-4^{n-1}}{(d^{n-1}+1)(d^{n-1}+2)(d^{n-1}+3)}.
\end{equation}
We note that even if \cref{even_local} is negative for $n\geq3$, the quantity scales like $O(d^{3(1-n)})$ so that the negative gap is effectively negligible for all practical purposes.

\subsection{[[4,2,2]] code}
\label{422}
Let us consider a CSS stabilizer code \cite{Nielsen_Chuang_2010}, which encodes two logical qubits in four. It is the smallest qubit stabilizer code which can detect a single qubit error, and moreover it is the smallest example of a toric code \cite{eczoo_stab_4_2_2}. The codewords are
\begin{align}
    |\overline{00}\rangle &= \frac{1}{\sqrt{2}}\Big(|0000\rangle + |1111\rangle\Big)\\
 |\overline{01}\rangle &= \frac{1}{\sqrt{2}}\Big(|0011\rangle + |1100\rangle\Big)\\
  |\overline{10}\rangle &= \frac{1}{\sqrt{2}}\Big(|0101\rangle + |1010\rangle\Big)\\
   |\overline{11}\rangle &= \frac{1}{\sqrt{2}}\Big(|0110\rangle + |1001\rangle\Big).
\end{align}
 The isotropic set $\mathcal{S}$ corresponding to this code is
\begin{align}
    \mathcal{S} = \Big\{&(0, 0, 0, 0, 0, 0, 0, 0), (0, 1, 0, 1, 0, 1, 0, 1),\\&(1, 0, 1, 0, 1, 0, 1, 0), (1, 1, 1, 1, 1, 1, 1, 1)\Big\}.\nonumber
\end{align}
As observed in \cref{zerogapsec}, the cardinality of $A_\mathcal{S}=\mathcal{S}^\perp$ is  $d_Sd_B=64$.

The ASE for this subspace is exactly $3/7$ which is the same as the ASE for a four-dimensional Hilbert space treated as two qubits: the magic gap is 0.
On the other hand, considering the encoded subspace as a four dimensional qudit, the magic gap is $-2/35$. Finally, we note that the first two basis states $\{|\overline{00}\rangle , |\overline{01}\rangle\}$ form a [[4,1,2]]  subcode. The ASE for the corresponding subspace is exactly $1/5$, the same as the ASE for a qubit, again giving zero magic gap.

\section{Numerical methods}
\label{numerical}

\begin{figure}[b]
\label{fig:monte_carlo}
\centering
\includegraphics[width=0.75\textwidth]{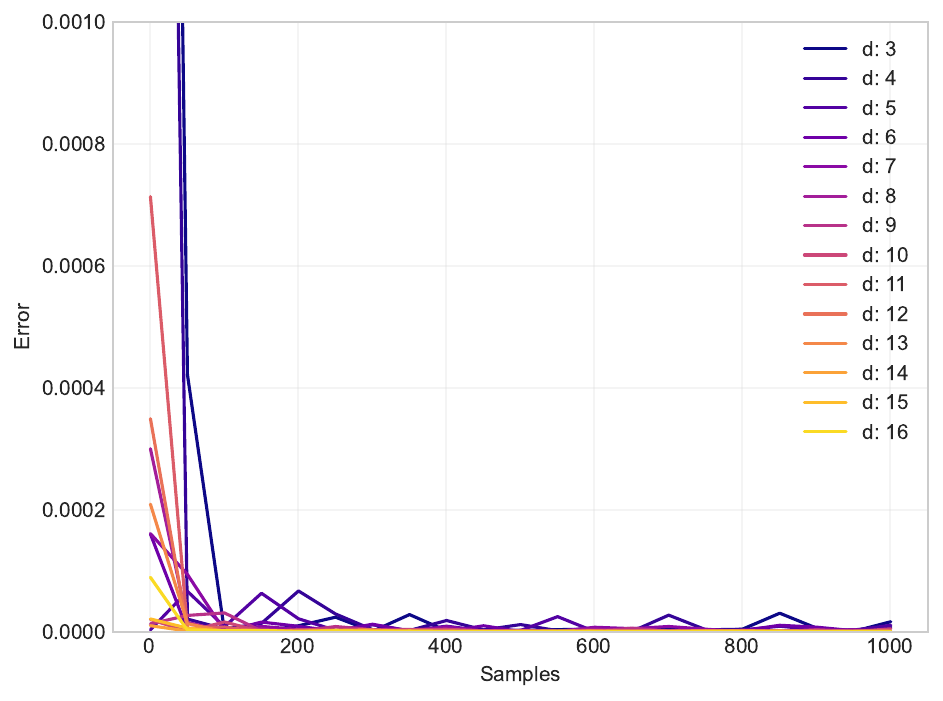}
\caption{We consider big Hilbert spaces of dimensions $3$ to $16$. For each, we select at random a $d_B-1$ dimensional subspace. We calculate the ASE exactly in each case, as well as approximate it using increasing numbers of Haar random samples. The error, given by $(\text{exact}-\text{mc})^2$, is plotted above.}.
\end{figure}

Using the results of Appendix \ref{App_Average}, we have developed an efficient \texttt{python} algorithm for calculating the average stabilizer entropy of an arbitrary subspace with respect to the (multi)qudit Weyl-Heisenberg group on the larger Hilbert space. On a standard laptop, it can calculate with respect to the five qubit WH group in approximately eleven minutes. Nevertheless, for calculations in larger dimensions and in order to use this quantity in optimization procedures, more efficient methods are necessary. To this end, we have also implemented a Monte Carlo approximation of the ASE of a subspace: we sample some number $M$ of Haar random states and average the stabilizer entropy over them. To evaluate the efficacy of this approximation, we plot in Figure \ref{fig:monte_carlo} the error in approximation for a variety of dimensions, and for increasing numbers of samples. 

The \texttt{python} library we have developed (\texttt{magicgap}) also contains useful routines for dealing with  all of the constructions used in this paper: totally isotropic subspaces, symmetrized qubits, the Majorana stellar representation, analyzing and visualizing quantum polyhedra, and extremizing the ASE of a subspace. It is open source and freely available at \cite{magicgap}.

\section{Conclusion}

In this paper, we have explored the interplay between the resource theory of nonstabilizerness and subspace embeddings. The embedding of one quantum system within another is ubiquitous, whether in the context of error correcting codes, the quantum simulation of qudits with $n$-qubits, or in symmetry constrained systems. The average stabilizer entropy (ASE) gap diagnoses the average difference in nonstabilizer resources between preparing a state in a chosen subspace and preparing the same state as such. We have provided explicit formulas for the ASE of an arbitrary Hilbert space, as well as the ASE of a subspace expressed in terms of the characteristic function of the corresponding projector. The latter, in particular, provides a tractable basis for numerical computation. Indeed, we have provided an efficient open source implementation of this paper's constructions, which we hope will be useful to researchers in quantum computation and beyond.

 We have proven that averaging the ASE of a subspace over all subspaces of fixed dimension is, in fact, the ASE of the larger Hilbert space, although the variance differs in the two cases. This allows us to see that in particular for $n$-qubit systems, as the subspace dimension approaches the dimension of the larger Hilbert space, the magic gap generically goes negative. We have also performed a exhaustive numerical search for subspaces which extremize the ASE for a variety of Hilbert space structures. We find that zero magic gap is attainable in general when the dimension of the subspace divides the overall dimension, and for the even multiqudit case, it is generally possible to achieve negative magic gap for certain combinations of dimensions.
 
Next, we have provided a sufficient condition for the ASE gap to be less than or equal to zero, focusing on the case of stabilizer codes with trivial group homomorphisms, and show that a) for $m$-qubits encoded in $n$-qubits, the ASE gap is 0, b) if $d$ is odd, then the ASE gap is also 0, and c) if $d$ is even, $d_S$ is even, and the isotropic set $\mathcal{S}$ defining the codespace is even, this implies zero magic gap for a single qudit, and negative magic gap for multiple qudits. To handle stabilizer codes beyond these two special cases, one must calculate the cardinality of the set $A_{\mathcal{S}}$. In the case of a $\mathbb{Z}_d$-invariant subspace, the cardinality of this set can be calculated exactly. Identifying additional cases in which the cardinality of this set can be derived from first principles would extend the current framework. Likewise, characterizing sufficient—and ideally also necessary—conditions for the ASE gap to be nonpositive beyond the setting of stabilizer codes remains an important direction for further development. We have also explored the effect of adding support on the complement of a given subspace, both with fixed and vector-dependent choices. While fixed support does not improve ASE when the subspace is already optimal, it can significantly reduce ASE for generic subspaces. These observations could be extended by optimizing not only over the subspace itself, but also over complementary subspaces to minimize average ASE.

Some numerical data still reveal unexpected patterns that need further theoretical explanation.
It is unclear why, in the single-qudit case, when the global Hilbert space has dimension a power of two, zero magic gap subspaces of dimension three consistently emerge.
Similarly, for the multiqubit and multiququart case, we can find a zero magic subspace whose dimension is 3 or 6. Indeed, we have also shown that for $d=2,n=2$, the symmetric subspace is an example of a zero magic gap subspace with $d_S=3$. One would be tempted to think that all multiples of three would have zero magic gap, but this does not hold since e.g.\ for four qubits, subspaces with dimensions $d_S=9,12,15$ have negative magic gap.

More broadly, a systematic investigation of the resource costs associated with embedding is needed. While selecting a subspace with minimal ASE minimizes nonstabilizer resources on average, its practical significance depends on how representative such averages are, given that most computations do not explore the entire Hilbert space. Additionally, the overhead of implementing encoding and decoding operations, as well as preparing highly entangled states, may offset the savings gained.  We conjecture that the amount of magic needed to perform the embedding is greater than or equal to the magic gained. In classical simulation contexts, the increased dimensionality likely outweighs reductions in sample complexity, though a formal proof remains open.
 
Finally, we have considered in this work an essentially static picture.
Introducing time evolution would require a measure of non-Cliffordness of unitaries to quantify this average quantity on a subspace. 
Exploring the relationship between minimizing this dynamic quantity and minimizing the ASE, as well as the dependence of the rate of magic production on the choice of subspace, constitute important questions for future works.

\begin{acknowledgments}
Matthew B. Weiss and Gianluca Cuffaro acknowledge support by the National Science Foundation through grants NSF-2210495 and OSI-2328774. Simone Cepollaro acknowledges support from the National Quantum Science and Technology Institute (NQSTI), PNRR MUR project PE0000023-NQSTI.

Special thanks to Blake Stacey for many helpful comments on this work.
\end{acknowledgments}

\bibliographystyle{utphys3}

\bibliography{bib}

\newpage
\appendix

\section{Qudit Extrema}
Here we show the results of extremizing the ASE for $d_B=3$ to $16$, treated as a single qudit, and for each choice of subspace dimension $d_S$.

\label{qubit_extrema}
\begin{figure}[H]
\label{fig:qudit_qudit}
\centering
\includegraphics[width=1\textwidth]{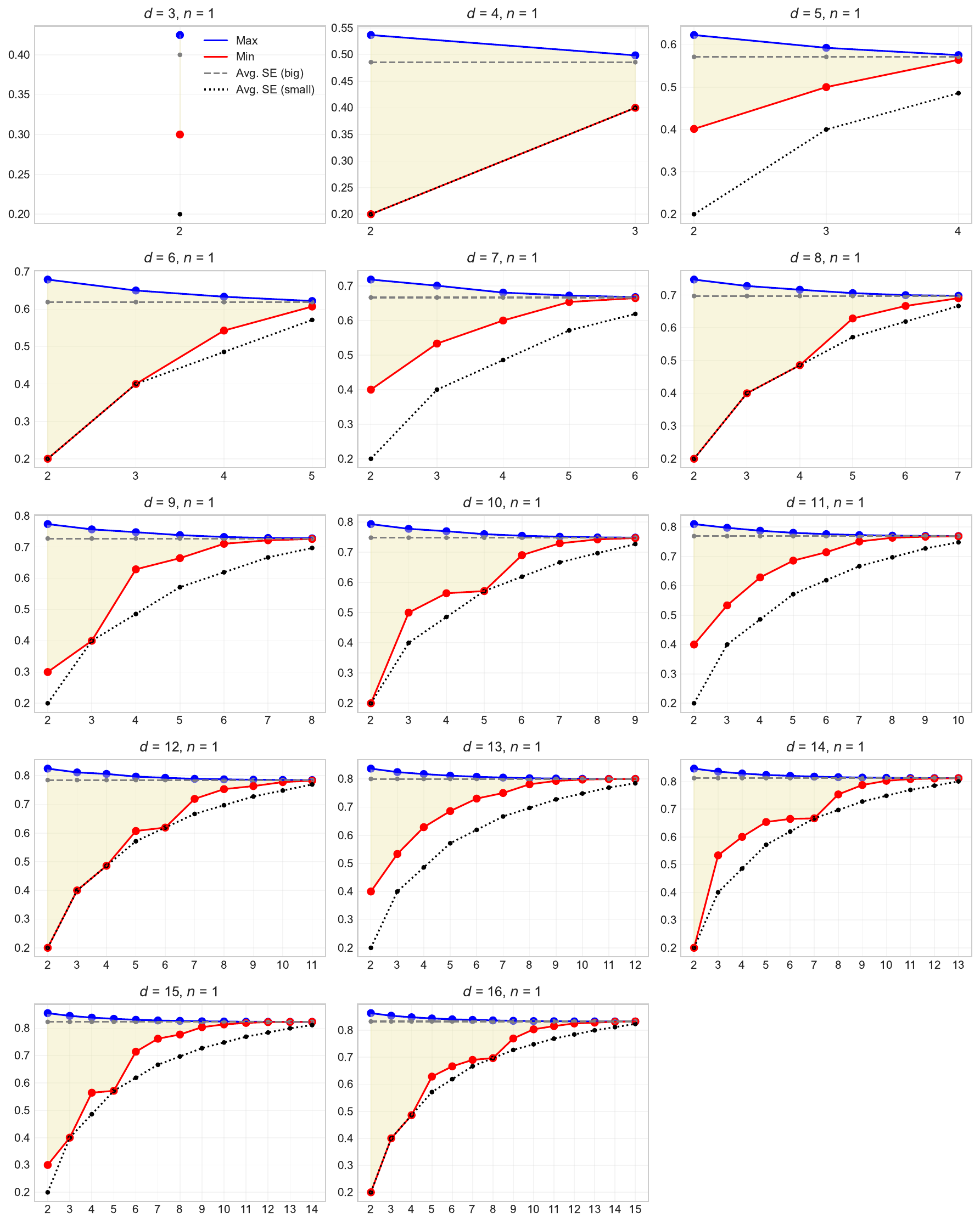}
\end{figure}

\section{Bounding the robustness of magic}
\label{RobustnessBounds}
In the spirit of \cite{Leone_2022, Heinrich2019}, let $\mathcal{D}(\psi)=\frac{1}{d^n}	\sum_{\B{a}}|\trace(D_{\B{a}}^\dagger \psi)|$. It is not hard to show that $D(\psi)=1$ for a stabilizer state, and this quantity satisfies \emph{convexity},
\begin{align}
\mathcal{D}\left(\sum_i x_i\sigma_i	\right)\leq \sum_i |x_i|D(\sigma_i).
\end{align}
 If $\psi=\sum_i x_i\sigma_i$ is an expansion in terms of stabilizer states, we have that  $\forall i: D(\sigma_i)=1$, and so $\frac{1}{d^n}	\sum_{\B{a}}|\trace(D_{\B{a}}^\dagger \psi)|\leq \mathcal{R}(\psi)$. As in the main body of the text, let $P_{\B{a}}(\psi)=\frac{1}{d^n}|\trace(D_{\B{a}}^\dagger \psi)|^2$, so that the R\'enyi stabilizer entropies may be defined as
\begin{align}
M_\alpha(\psi) =  \frac{1}{1-\alpha}\log \left(\sum_{\B{a}} P_{\B{a}}(\psi)^{\alpha}\right) - \log d^n.
\end{align}
In particular,
\begin{align}
M_{\frac{1}{2}}(\psi)&=2\log \left(\sum_i \frac{1}{d^{n/2}}|\trace(D_{\B{a}}^\dagger \psi)|\right)-\log d^n\\
&=2\log \left(\sum_i \frac{1}{d^{n}} |\trace(D_{\B{a}}^\dagger \psi)|\right) \leq 2\log(\mathcal{R}(\psi)).
\end{align}
Since for $0 \leq \alpha_1 \leq \alpha_2 \leq \infty$, the R\'enyi entropies satisfy $H_{\alpha_2} \leq H_{\alpha_1}$, we conclude that for $\alpha \ge 1/2$,
\begin{align}
M_\alpha(\psi) \leq 2 \log(\mathcal{R}(\psi)).
\end{align}
Meanwhile, the 2-Tsallis stabilizer entropy is
\begin{align}
M(\psi) = 1- \frac{1}{d^n}\sum_{\B{a}}|\trace(D_{\B{a}}^\dagger \psi)|^4,
\end{align}
which we can relate to the 2-R\'enyi entropy as
\begin{align}
M_2(\psi)&=-\log \left(	\frac{1}{d^{2n}}\sum_{\B{a}}|\trace(D_{\B{a}}^\dagger \psi)|^4\right) - \log d^n=-\log \left(	\frac{1}{d^{n}}\sum_{\B{a}}|\trace(D_{\B{a}}^\dagger \psi)|^4\right)\\
&=-\log\left(1-M(\psi)\right).
\end{align}
Plugging this into the bound, we arrive at $M(\psi)	\leq 1-\frac{1}{\mathcal{R}(\psi)^2}$. To get the final result, we note that $M_2(\psi) \leq -\log\frac{2}{d+1}$, so that $M(\psi)\leq 1-\frac{2}{d+1}<1$. Moreover $\mathcal{R}(\psi)> 0$. Thus we can express our lower bound on the robustness of magic in terms of the linear stabilizer entropy as
\begin{align}
		\mathcal{R}(\psi)^2\ge \frac{1}{1-M(\psi)}.
\end{align}

\section{Weingarten calculus}
\label{App_Weingarten}
Here we review some basic details regarding integration over the unitary group $U(d)$ with respect to the Haar measure. Integrals of polynomial type can be computed explicitly using the \emph{Weingarten calculus} \cite{Collins2006,Collins_2022}, a combinatorial technique that expresses unitary integrals as weighted sums over permutations. Let \( \mathcal{H} \cong \mathbb{C}^d \) be a Hilbert space of dimension \( d \), and consider a bounded operator \( O \in \mathcal{B}(\mathcal{H}^{\otimes n}) \). The average of the conjugated operator \( U^{\otimes n} O \, (U^\dagger)^{\otimes n} \) over the unitary group \( U(d) \), with respect to the normalized Haar measure \( dU \), admits the exact expression \cite{Leone2021quantumchaosis}
\begin{equation}
\int_{U(d)} dU \, (U^\dagger)^{\otimes n} \, O \, U^{\otimes n}
= \sum_{\pi, \sigma \in S_n} \mathrm{Wg}(\pi,\sigma) \, \mathrm{tr}(O T_\sigma) \, T_\pi,
\label{eq:weingarten_operator_form}
\end{equation}
where \( S_n \) is the symmetric group on \( n \) elements, \( T_\pi \) is the unitary representation of the permutation \( \pi \) on \( \mathcal{H}^{\otimes n} \), and \( \mathrm{Wg}(\pi, \sigma) \) are the Weingarten coefficients given by
\begin{equation}
\mathrm{Wg}(\pi, \sigma) = \sum_{\lambda \vdash n} \frac{d_\lambda^2}{(n!)^2} \, \frac{\chi_\lambda(\pi \sigma)}{D_\lambda}.
\label{eq:weingarten_formula}
\end{equation}
Here \( \lambda \vdash n \) runs over all partitions of \( t \), \( d_\lambda \) is the dimension of the irreducible representation of \( S_n \) labeled by \( \lambda \), \( \chi_\lambda \) is the corresponding character, and \( D_\lambda = s_\lambda(u)\) is the Schur polynomial evaluated at \( u=(1,\dots,1) \in \mathbb{C}^d \). 

In particular, we are interested in evaluating 
\be \sum_{\sigma}\text{Wg}(\pi,\sigma). \ee
Using \eqref{eq:weingarten_formula}, we have
\ba 
\sum_{\sigma} \text{Wg}(\pi,\sigma) & = & \sum_{\sigma}\sum_{\lambda\vdash n} \frac{d_{\lambda}^2}{(n!)^2}\frac{\chi_{\lambda}(\pi \sigma)}{D_{\lambda}} \\
& = & \sum_{\lambda}\frac{d_{\lambda}^2}{(n!)^2 D_{\lambda}} \sum_{\sigma}\chi_{\lambda}(\pi\sigma) .
\ea
Using the relation $\sum_{\sigma} \chi_{\lambda}(\sigma)=\delta_{\lambda,n}n!$, and the fact that $d_n=1$ and $D_n=\binom{d+n-1}{n}$, we arrive at
\be 
\sum_{\sigma} \text{Wg}(\sigma,\tau) = \frac{1}{n!}\binom{d+n-1}{n}^{-1} .
\ee

\section{Averaging the stabilizer entropy}
\label{App_Average}
In order to calculate both the intrinsic and extrinsic average stabilizer entropies, we must evaluate an expression of the form
\begin{align}
\trace\left(Q\Pi_{\mathcal{E}}^{\otimes 4}\Pi_{\text{sym}^4} \right),
\end{align}
where
\ba 
Q &=& \frac{1}{(d^n)^2}\sum_\textbf{a} D_\textbf{a} \ot D_\textbf{a}^{\dag}\ot D_\textbf{a}\ot D_\textbf{a}^{\dag} \\
\Pi_{\text{sym}^4}&=&\frac{1}{4!}\sum_{\pi \in S_4} T_{\pi} =\frac{1}{4!}\sum_{\pi \in S_4}\sum_{i,j,k,l=0}^{d-1}\ket{i,j,k,l}\bra{\pi(i,j,k,l)},
\ea 
for WH operators $\{D_{\B{a}}\}$ acting on $\mathcal{H}_d^{\otimes n}$. By direct substitution,
\begin{align}
\trace\left(Q\Pi_{\mathcal{E}}^{\otimes 4}\Pi_{\text{sym}^4} \right) &= \frac{1}{24(d^n)^2}  \sum_{\pi \in S_4} \sum_{\textbf{a}} \trace\Big(\big(D_\textbf{a}\Pi_{\mathcal{E}} \ot D_\textbf{a}^{\dag} \Pi_{\mathcal{E}} \ot D_\textbf{a} \Pi_{\mathcal{E}} \ot D_\textbf{a}^{\dag}\Pi_{\mathcal{E}}\big)  T_\pi\Big).
\end{align}
We begin by recalling that for a permutation operator $T_\pi$,
\ba 
\trace\Big((X_1\ot \ldots \ot X_n) T_{\pi}\Big)=\prod_{c\in \chi(\pi)}\trace\left(\prod_{i \in c}X_i\right), \ea
where $\chi(\pi)$ denotes the set of cycles of the permutation $\pi$, and $c$ a particular cycle. For example, suppose $\pi=(12)(3)(4)$. Then $\chi(\pi)=\{(12),(3),(4)\}$ and
\be 
\trace\Big((X_1\ot X_2 \ot X_3 \ot X_4) T_{\pi}\Big) =\trace(X_1 X_2)\trace(X_3)\trace(X_4).
\ee 
We shall work permutation by permutation and evaluate 
\begin{align}
\sum_{\textbf{a}} \trace\Big(\big(D_\textbf{a}\Pi_{\mathcal{E}} \ot D_\textbf{a}^{\dag} \Pi_{\mathcal{E}} \ot D_\textbf{a} \Pi_{\mathcal{E}} \ot D_\textbf{a}^{\dag}\Pi_{\mathcal{E}}\big)  T_\pi\Big)	
\end{align}
in each case. Ultimately, we seek an expression in terms of the characteristic function of the projector $\Pi_{\mathcal{E}}$, that is,
\begin{align}
\chi_\B{a}\equiv \chi_{\B{a}}(\Pi_{\mathcal{E}}) = \frac{1}{d^n}\trace(D_{\B{a}}^\dagger \Pi_{\mathcal{E}}) && \Pi_{\mathcal{E}} = \sum_{\B{a}}\chi_\B{a}D_{\B{a}}=\Pi_{\mathcal{E}} ^\dagger =  \sum_{\B{a}}\chi_\B{a}^*D_{\B{a}}^\dagger,
\end{align}
and we shall make use of the fact that
\begin{align}
D_{\B{a}}D_{\B{b}} = \tau^{[\B{a},\B{b}]}D_{\B{a}+\B{b}}
\end{align}
where $\tau = -e^{i\pi/d}$.

Before moving to the general case, we note that the intrinsic stabilizer entropy corresponds to the special case that $\Pi_{\mathcal{E}}=I$. It is straightforward to see that
\begin{itemize}
    \item  $\pi=(1)(2)(3)(4) \rightarrow  \sum_{\textbf{a}}  \trace(D_{\textbf{a}}) \trace (D_{\textbf{a}}^{\dag}) \trace (D_{\textbf{a}} )\trace (D_{\textbf{a}}^{\dag})=(d^n)^4 $,
    \item $\pi=(12)(3)(4)$ [all 6 swaps equal] $\rightarrow \sum_{\textbf{a}}  \trace (D_{\textbf{a}}D_{\textbf{a}}^{\dag})  \trace (D_{\textbf{a}}) \trace (D_{\textbf{a}}^{\dag}) =(d^n)^3  $,
    \item $\pi=(123)(4)$ [all 8 3-cycles equal] $\rightarrow  \sum_{\textbf{a}}  \trace (D_{\textbf{a}}D_{\textbf{a}}^{\dag} D_{\textbf{a}})  \trace (D_{\textbf{a}}^{\dag}) =(d^n)^2 $,
     \item $\pi=(1234)$ [all 6 4-cycles equal] $\rightarrow    \sum_{\textbf{a}}  \trace (D_{\textbf{a}}D_{\textbf{a}}^{\dag} D_{\textbf{a}}  D_{\textbf{a}}^{\dag} ) =(d^n)^3  $,
    \item $\pi\in\{(12)(34),(14)(23)\}$ $\rightarrow \sum_{\textbf{a}}  \trace (D_{\textbf{a}}D_{\textbf{a}}^{\dag}) \trace (D_{\textbf{a}}D_{\textbf{a}}^{\dag}) = (d^n)^4$,
     \item $\pi=(13)(24)\rightarrow  \sum_{\textbf{a}}  \trace (D_{\textbf{a}}^2) \trace ({D_{\textbf{a}}^{\dag2}}) =\sum_{\textbf{a}} (d^n)^2 \delta_{2{\textbf{a}},0} = \begin{cases} (d^n)^2, & d \text{ odd (multiqudit)} \\ 4(d^n)^2, &d \text{ even (multiqudit)} \\ (d^n)^4, & d=2 \text{ (multiqubit)} ,\end{cases}$
\end{itemize}
where the last case follows from the fact that for $n$-qubits, the $D_\textbf{a}$'s are always hermitian, corresponding to strings of Pauli operators, and $\trace{D_\textbf{a}^2}=d^n$ for any Pauli string. Thus
\begin{align}
	\trace\left(Q \Pi_{\text{sym}^4} \right)
&=\begin{cases} \frac{1}{8}(d^n+1)(d^n+3), & d \text{ odd (multiqudit)} \\ \frac{1}{8}(d^n+2)^2, & d \text{ even (multiqudit)} \\ \frac{1}{6}(d^n+1)(d^n+2), & d=2 \text{ (multiqubit)},\end{cases}
\end{align}
which is in agreement with the previously known results for multiqubit case \cite{Leone_2022}.

We now turn to the general case. Permutation by permutation, we have:

\begin{enumerate}
    \item $\pi=(1)(2)(3)(4)$ 
    \begin{align}
\sum_\B{a}  \trace(D_\B{a}\Pi_{\mathcal{E}}) \trace(D_\B{a}^\dagger \Pi_{\mathcal{E}}) \trace(D_\B{a}\Pi_{\mathcal{E}}) \trace(D_\B{a}^\dagger\Pi_{\mathcal{E}})	= (d^n)^4\sum_{\B{a}} |\chi_\B{a}|^4.
\end{align}

 \item $\pi\in\{(1)(2)(34), (1)(23)(4), (12)(3)(4), (14)(2)(3) \}$
    
    These permutations lead to
\begin{align}
\sum_{\B{a}}  \trace(D_{\B{a}}\Pi_{\mathcal{E}}D_{\B{a}}^{\dag}\Pi_{\mathcal{E}})  \trace(D_{\B{a}}\Pi_{\mathcal{E}}) \trace (D_{\B{a}}^{\dag}\Pi_{\mathcal{E}}).
\end{align}
We first note that
\begin{align}
D_\B{a}\Pi_{\mathcal{E}} D_\B{a}^\dagger &= D_\B{a}\left(\sum_\B{b}\chi_\B{b}D_\B{b}\right)D_\B{a}^\dagger =\sum_\B{b} \chi_\B{b}D_\B{a} D_\B{b}D_{-\B{a}}\\
&=	\sum_\B{b} \chi_\B{b}\tau^{[\B{a},\B{b}]}D_{\B{a}+\B{b}}D_{-\B{a}}=\sum_\B{b} \chi_\B{b}\tau^{[\B{a},\B{b}]}\tau^{[\B{a}+\B{b},-\B{a}]}D_{\B{b}}.
\end{align}
Since $[\B{a}+\B{b},-\B{a}]=[\B{a},\B{b}]$ and $\tau^2 = \omega$,
\begin{align}
D_\B{a}\Pi_{\mathcal{E}} D_\B{a}^\dagger=\sum_\B{b} \chi_\B{b}\omega^{[\B{a},\B{b}]}D_{\B{b}}.
\end{align}
Then
\begin{align}
	\trace(D_{\B{a}}\Pi_{\mathcal{E}}D_{\B{a}}^{\dag}\Pi_{\mathcal{E}})=& \trace\left(\sum_\B{b} \chi_\B{b}\omega^{[\B{a},\B{b}]}D_{\B{b}}\sum_{\B{c}}\chi_\B{c}^*D_{\B{c}}^\dagger\right)\\
	&=\sum_{\B{b}, \B{c}} \chi_\B{b}\omega^{[\B{a}, \B{b}]}\chi_\B{c}^*\trace(D_\B{b}D_\B{c}^\dagger)\\
	&=\sum_{\B{b}, \B{c}} \chi_\B{b}\omega^{[\B{a}, \B{b}]}\chi_\B{c}^* d\delta_{\B{b}, \B{c}}=d\sum_\B{b} |\chi_\B{b}|^2\omega^{[\B{a}, \B{b}]},
\end{align}
so that 
\begin{align}
\sum_{\B{a}}  \trace(D_{\B{a}}\Pi_{\mathcal{E}}D_{\B{a}}^{\dag}\Pi_{\mathcal{E}})  \trace(D_{\B{a}}\Pi_{\mathcal{E}}) \trace (D_{\B{a}}^{\dag}\Pi_{\mathcal{E}})	&=\sum_{\B{a}}  \left(d\sum_\B{b} |\chi_\B{b}|^2\omega^{[\B{a}, \B{b}]}\right)  d^2|\chi_\B{a}|^2\\
&=(d^n)^3\sum_{\B{a}, \B{b}}  |\chi_\B{a}|^2  |\chi_\B{b}|^2 \omega^{[\B{a}, \B{b}]}.
\end{align}
    
   \item $\pi= (1)(24)(3)$
   
  This permutation yields
\begin{align}
\sum_{\B{a}}  \trace(D_{\B{a}}^\dagger\Pi_{\mathcal{E}}D_{\B{a}}^{\dag}\Pi_{\mathcal{E}})  \trace(D_{\B{a}}\Pi_{\mathcal{E}}) \trace (D_{\B{a}}\Pi_{\mathcal{E}})	.
\end{align}
We note that
\begin{align}
D_\B{a}^\dagger\Pi_{\mathcal{E}} D_\B{a}^\dagger &= D_\B{a}^\dagger\left(\sum_\B{b}\chi_\B{b}D_\B{b}\right)D_\B{a}^\dagger =\sum_\B{b} \chi_\B{b}D_{-\B{a}} D_\B{b}D_{-\B{a}}\\
&=	\sum_\B{b} \chi_\B{b}\tau^{[-\B{a},\B{b}]}D_{-\B{a}+\B{b}}D_{-\B{a}}=\sum_\B{b} \chi_\B{b}\tau^{[-\B{a},\B{b}]}\tau^{[-\B{a}+\B{b},-\B{a}]}D_{\B{b}-2\B{a}}.
\end{align}
But  $[-\B{a}+\B{b},-\B{a}]=[\B{a},\B{b}]$ so that
\begin{align}
D_\B{a}^\dagger\Pi_{\mathcal{E}} D_\B{a}^\dagger=\sum_\B{b} \chi_\B{b}D_{\B{b}-2\B{a}}.
\end{align}
Then
\begin{align}
	\trace(D_{\B{a}}^\dagger\Pi_{\mathcal{E}}D_{\B{a}}^{\dag}\Pi_{\mathcal{E}}) &= \trace\left(\sum_\B{b} \chi_\B{b}D_{\B{b}-2\B{a}} \sum_{\B{c}}\chi_\B{c}^*D_{\B{c}}^\dagger\right)\\
	&=\sum_{\B{b}, \B{c}}\chi_\B{b}\chi^*_\B{c}  d\delta_{\B{b}-2\B{a},\B{c}}= d\sum_{\B{b}}\chi_\B{b}\chi_{\B{b}-2\B{a}}^*,
\end{align}
which allows us to arrive at
\begin{align}
	\sum_{\B{a}}  \trace(D_{\B{a}}^\dagger\Pi_{\mathcal{E}}D_{\B{a}}^{\dag}\Pi_{\mathcal{E}})  \trace(D_{\B{a}}\Pi_{\mathcal{E}}) \trace (D_{\B{a}}\Pi_{\mathcal{E}}) &= \sum_{\B{a}}  \left(d\sum_{\B{b}}\chi_\B{b}\chi_{\B{b}-2\B{a}}^*\right)d^2\chi_{\B{a}}^{*2}\\
	&=(d^n)^3\sum_{\B{a}, \B{b}}\chi_{\B{a}}^{*2}\chi_\B{b}\chi_{\B{b}-2\B{a}}^*\\
    &= (d^n)^3\sum_{\textbf{a}, \textbf{b}}\chi_{\textbf{a}}^{2}\chi_\textbf{b}\chi_{\textbf{b}+2\textbf{a}}^*.
\end{align} 
   
  \item $\pi= (13)(2)(4)$
  
  This permutation delivers
\begin{align}
\sum_{\textbf{a}}  \trace(D_{\textbf{a}}\Pi_{\mathcal{E}}D_{\textbf{a}}\Pi_{\mathcal{E}})  \trace(D_{\textbf{a}}^\dagger\Pi_{\mathcal{E}}) \trace (D_{\textbf{a}}^\dagger\Pi_{\mathcal{E}}),	
\end{align}
but since we are summing over all $\B{a}$, this must be the same as (3).
    
    \item $\pi\in\{(1)(234),(1)(243), (123)(4), (124)(3), (132)(4), (134)(2), (142)(3), (143)(2)\}$
    
    These permutations bring us
\begin{align}
\sum_{\B{a}}  \trace(D_{\B{a}}\Pi_{\mathcal{E}}D_{\B{a}}^{\dag}\Pi_{\mathcal{E}} D_{\B{a}} \Pi_{\mathcal{E}}) \trace(D_{\B{a}}^{\dag} \Pi_{\mathcal{E}})	.
\end{align}
On the one hand, we've seen that $
D_\B{a}\Pi_{\mathcal{E}} D_\B{a}^\dagger=\sum_\B{b} \chi_\B{b}\omega^{[\B{a},\B{b}]}D_{\B{b}}$. On the other hand,
\begin{align}
\Pi_\mathcal{E} D_\B{a} \Pi_\mathcal{E} &= 	\left(\sum_{\B{c}}\chi_\B{c}D_{\B{c}}\right)D_\B{a}\left(\sum_{\B{d}}\chi_\B{d}D_{\B{d}}\right)\\
&=\sum_{\B{c},\B{d}} \chi_\B{c}\chi_\B{d}D_{\B{c}} D_\B{a}D_{\B{d}}=\sum_{\B{c},\B{d}} \chi_\B{c}\chi_\B{d} \tau^{[\B{c}, \B{a}]}D_{\B{c}+\B{a}}D_{\B{d}}\\
&=\sum_{\B{c},\B{d}} \chi_\B{c}\chi_\B{d} \tau^{[\B{c}, \B{a}]}\tau^{[\B{c}+\B{a},\B{d}]}D_{\B{a}+\B{c}+\B{d}} = \sum_{\B{c},\B{d}} \chi_\B{c} \chi_\B{d} \tau^{-[\B{a}, \B{c}]+[\B{a},\B{d}]+[\B{c},\B{d}]}D_{\B{a}+\B{c}+\B{d}}.
\end{align}
Thus
\begin{align}
	\trace(D_{\B{a}}\Pi_{\mathcal{E}}D_{\B{a}}^{\dag}\Pi_{\mathcal{E}} D_{\B{a}} \Pi_{\mathcal{E}}) &= \trace\left(\sum_\B{b} \chi_\B{b}\omega^{[\B{a},\B{b}]}D_{\B{b}}\sum_{\B{c},\B{d}} \chi_\B{c}\chi_\B{d} \tau^{-[\B{a}, \B{c}]+[\B{a},\B{d}]+[\B{c},\B{d}]}D^\dagger_{-\B{a}-\B{c}-\B{d}} \right)\\
	&=\sum_{\B{b}, \B{c}, \B{d}} \chi_\B{b}\chi_\B{c}\chi_\B{d} \omega^{[\B{a},\B{b}]}\tau^{-[\B{a}, \B{c}]+[\B{a},\B{d}]+[\B{c},\B{d}]}d \delta_{\B{b}, -\B{a}-\B{c}-\B{d}}.
\end{align}
This requires $\B{d}=-\B{a}-\B{b}-\B{c}$, and so
\begin{align}
	-[\B{a}, \B{c}]+[\B{a},\B{d}]+[\B{c},\B{d}]&= -[\B{a}, \B{c}]+[\B{a}, -\B{a}-\B{b}-\B{c}]+[\B{c}, -\B{a}-\B{b}-\B{c}]\\
	&=-[\B{a}, \B{c}]-[\B{a}, \B{b}]-[\B{a},\B{c}]-[\B{c},\B{a}]-[\B{c},\B{b}]\\
	&=-[\B{a}, \B{b}]-[\B{a}, \B{c}]-[\B{c},\B{b}],
\end{align}
which leads to
\begin{align}
\trace(D_{\B{a}}\Pi_{\mathcal{E}}D_{\B{a}}^{\dag}\Pi_{\mathcal{E}} D_{\B{a}} \Pi_{\mathcal{E}}) &= d^n\sum_{\B{b},\B{c}}\chi_\B{b}\chi_\B{c}\chi^*_{\B{a}+\B{b}+\B{c}}\tau^{[\B{a}, \B{b}]-[\B{a}, \B{c}]-[\B{c}, \B{b}]},
\end{align}
from which we conclude
\begin{align}
\sum_{\B{a}}  \trace(D_{\B{a}}\Pi_{\mathcal{E}}D_{\B{a}}^{\dag}\Pi_{\mathcal{E}} D_{\B{a}} \Pi_{\mathcal{E}}) \trace(D_{\B{a}}^{\dag} \Pi_{\mathcal{E}})	&= (d^n)^2\sum_{\B{a}, \B{b}, \B{c}}\chi_\B{a}\chi_\B{b}\chi_\B{c}\chi^*_{\B{a}+\B{b}+\B{c}}\tau^{[\B{a}, \B{b}]-[\B{a}, \B{c}]-[\B{c}, \B{b}]}.
\end{align}

     \item $\pi\in\{(1234), (1432) \}$ 
    
    These permutations yield
\begin{align}
 &\sum_{\textbf{a}}  \trace(D_{\textbf{a}}\Pi_{\mathcal{E}}D_{\textbf{a}}^{\dag}\Pi_{\mathcal{E}} D_{\textbf{a}}\Pi_{\mathcal{E}}  D_{\textbf{a}}^{\dag}\Pi_{\mathcal{E}})\\
 &=\sum_\B{a}\trace\left(\sum_\B{b} \chi_\B{b}\omega^{[\B{a},\B{b}]}D_{\B{b}}\sum_{\B{c},\B{d}} \chi_\B{c} \chi_\B{d} \tau^{-[\B{a}, \B{c}]+[\B{a},\B{d}]+[\B{c},\B{d}]}D_{\B{a}+\B{c}+\B{d}} D_{\B{a}}^\dagger  \sum_{\B{e}}\chi_\B{e}D_{\B{e}}\right)\\
 &=\sum_{\B{a},\B{b}, \B{c}, \B{d},\B{e}}\chi_{\B{b}}\chi_\B{c}\chi_\B{d}\chi_\B{e}\omega^{[\B{a},\B{b}]}\tau^{-[\B{a}, \B{c}]+[\B{a},\B{d}]+[\B{c},\B{d}]}\trace(D_\B{b}D_{\B{a}+\B{c}+\B{d}} D_{\B{a}}^\dagger D_{\B{e}}),
\end{align}
where
\begin{align}
	\trace(D_\B{b}D_{\B{a}+\B{c}+\B{d}} D_{\B{a}}^\dagger D_{\B{e}}) &= \tau^{[\B{b}, \B{a}+\B{c}+\B{d}]}\tau^{[-\B{a}, \B{e}]}\trace(D_{\B{a}+\B{b}+\B{c}+\B{d}}D^\dagger_{\B{a}-\B{e}})\\
	&=\tau^{[\B{b}, \B{a}+\B{c}+\B{d}]-[\B{a}, \B{e}]}d\delta_{\B{a}+\B{b}+\B{c}+\B{d},\B{a}-\B{e}},
\end{align}
which forces $\B{e}=-\B{b}-\B{c}-\B{d}$ so that the argument of new phase becomes
\begin{align}
	[\B{b}, \B{a}+\B{c}+\B{d}]-[\B{a}, \B{e}]&=[\B{b}, \B{a}+\B{c}+\B{d}]-[\B{a}, -\B{b}-\B{c}-\B{d}]\\
	&=[\B{b}, \B{a}]+[\B{b}, \B{c}] + [\B{b}, \B{d}] + [\B{a}, \B{b}]+ [\B{a}, \B{c}]+[\B{a}, \B{d}]\\
	&=[\B{b}, \B{c}] + [\B{b}, \B{d}] + [\B{a}, \B{c}]+[\B{a}, \B{d}]\\
	&=[\B{a}+\B{b},\B{c} + \B{d}].
\end{align}
Thus
\begin{align}
 &\sum_{\textbf{a}}  \trace(D_{\textbf{a}}\Pi_{\mathcal{E}}D_{\textbf{a}}^{\dag}\Pi_{\mathcal{E}} D_{\textbf{a}}\Pi_{\mathcal{E}}  D_{\textbf{a}}^{\dag}\Pi_{\mathcal{E}})\\
&= d^n\sum_{\B{a},\B{b}, \B{c}, \B{d}}\chi_{\B{b}}\chi_\B{c}\chi_\B{d}\chi_\B{\B{b}+\B{c}+\B{d}}^*\tau^{2[\B{a},\B{b}]-[\B{a}, \B{c}]+[\B{a},\B{d}]+[\B{c},\B{d}]+[\B{a}+\B{b},\B{c} + \B{d}]}\\
&= d^n\sum_{\B{b}, \B{c}, \B{d}}\chi_{\B{b}}\chi_\B{c}\chi_\B{d}\chi_\B{\B{b}+\B{c}+\B{d}}^*\tau^{[\B{b},\B{c}+\B{d}]+[\B{c}, \B{d}]}\sum_\B{a}\omega^{[\B{a},\B{b}+\B{d}]}\\
&= d^n\sum_{\B{b}, \B{c}, \B{d}}\chi_{\B{b}}\chi_\B{c}\chi_\B{d}\chi_\B{\B{b}+\B{c}+\B{d}}^*\tau^{[\B{b},\B{c}+\B{d}]+[\B{c}, \B{d}]}d^2 \delta_{\B{b},-\B{d}}\\
&= (d^n)^3\sum_{\B{b}, \B{c}}\chi_{\B{b}}\chi_\B{c}\chi_{-\B{b}}\chi_{\B{b}+\B{c}-\B{b}}^*\tau^{[\B{b},\B{c}-\B{b}]+[\B{c}, -\B{b}]}\\
&= (d^n)^3\sum_{\B{b}, \B{c}}\chi_{\B{b}}\chi_\B{c}\chi_{\B{b}}^*\chi_{\B{c}}^*\tau^{[\B{b},\B{c}]-[\B{c}, \B{b}]}\\
&=(d^n)^3\sum_{\B{b}, \B{c}} |\chi_\B{b}|^2 |\chi_\B{c}|^2\omega^{[\B{b},\B{c}]},
 \end{align}
which shows that $(6)=(2)$. 

  \item $\pi\in\{(1243), (1342), (1324), (1423) \}$ 
   
These permutations deliver
\begin{align}
&\sum_{\textbf{a}}  \trace(D_{\textbf{a}}\Pi_{\mathcal{E}}D_{\textbf{a}}^{\dag}\Pi_{\mathcal{E}} D_{\textbf{a}}^{\dag}\Pi_{\mathcal{E}}  D_{\textbf{a}}\Pi_{\mathcal{E}})	\\
&=\sum_{\textbf{a}}  \trace\left( \sum_\B{b} \chi_\B{b}\omega^{[\B{a},\B{b}]}D_{\B{b}} \sum_\B{e} \chi_\B{e} D_\B{e}D_\B{a}^\dagger  \sum_{\B{c},\B{d}} \chi_\B{c} \chi_\B{d} \tau^{-[\B{a}, \B{c}]+[\B{a},\B{d}]+[\B{c},\B{d}]}D_{\B{a}+\B{c}+\B{d}}\right)\\
&=\sum_{\B{a}, \B{b}, \B{c}, \B{d}, \B{e}} \chi_\B{b}\chi_\B{c}\chi_\B{d}\chi_\B{e}\omega^{[\B{a},\B{b}]}\tau^{-[\B{a}, \B{c}]+[\B{a},\B{d}]+[\B{c},\B{d}]}\trace(D_\B{b}D_\B{e}D_\B{a}^\dagger D_{\B{a}+\B{c}+\B{d}}),
\end{align}
where
\begin{align}
	\trace(D_\B{b}D_\B{e}D_\B{a}^\dagger D_{\B{a}+\B{c}+\B{d}})&=\tau^{[\B{b}, \B{e}]}\tau^{[-\B{a}, \B{a}+\B{c}+\B{d}]}\trace(D_{\B{b}+\B{e}}D_{\B{c}+\B{d}})\\
	&=\tau^{[\B{b}, \B{e}]-[\B{a},\B{c}]-[\B{a}, \B{d}]}d \delta_{\B{b}+\B{e},-\B{c}-\B{d}}.
\end{align}
Thus $\B{e}=-\B{b}-\B{c}-\B{d}$. The argument to to the new phase factor becomes
\begin{align}
[\B{b}, -\B{b}-\B{c}-\B{d}]-[\B{a},\B{c}]-[\B{a},\B{d}] &= -[\B{b}, \B{c}]-[\B{b}, \B{d}]-[\B{a},\B{c}]-[\B{a},\B{d}].
\end{align}
Thus
\begin{align}
	&\sum_{\textbf{a}}  \trace(D_{\textbf{a}}\Pi_{\mathcal{E}}D_{\textbf{a}}^{\dag}\Pi_{\mathcal{E}} D_{\textbf{a}}^{\dag}\Pi_{\mathcal{E}}  D_{\textbf{a}}\Pi_{\mathcal{E}})\\
&=d^n\sum_{\B{a}, \B{b}, \B{c}, \B{d}} \chi_\B{b}\chi_\B{c}\chi_\B{d}\chi^*_{\B{b}+ \B{c}+\B{d}}\tau^{2[\B{a}, \B{b}]-[\B{a}, \B{c}]+[\B{a},\B{d}]+[\B{c},\B{d}]-[\B{b}, \B{c}]-[\B{b}, \B{d}]-[\B{a},\B{c}]-[\B{a},\B{d}]}\\
&=d^n\sum_{\B{a}, \B{b}, \B{c}, \B{d}} \chi_\B{b}\chi_\B{c}\chi_\B{d}\chi^*_{\B{b}+ \B{c}+\B{d}}\tau^{2[\B{a}, \B{b}-\B{c}]-[\B{b}, \B{c}+\B{d}]+[\B{c},\B{d}]}\\
&=d^n\sum_{\B{b}, \B{c}, \B{d}} \chi_\B{b}\chi_\B{c}\chi_\B{d}\chi^*_{\B{b}+ \B{c}+\B{d}}\tau^{-[\B{b}, \B{c}+\B{d}]+[\B{c},\B{d}]}\sum_{\B{a}}\omega^{[\B{a}, \B{b}-\B{c}]}\\
&=(d^n)^3\sum_{\B{b}, \B{c}, \B{d}} \chi_\B{b}\chi_\B{c}\chi_\B{d}\chi^*_{\B{b}+ \B{c}+\B{d}}\tau^{-[\B{b}, \B{c}+\B{d}]+[\B{c},\B{d}]}\delta_{\B{b}, \B{c}}\\
&=(d^n)^3\sum_{\B{b}, \B{d}} \chi_\B{b}\chi_\B{b}\chi_\B{d}\chi^*_{2\B{b}+\B{d}}\tau^{-[\B{b}, \B{b}+\B{d}]+[\B{b},\B{d}]}\\
&=(d^n)^3\sum_{\B{b}, \B{d}} \chi_\B{b}^2\chi_\B{d}\chi^*_{2\B{b}+\B{d}}\\
&=(d^n)^3\sum_{\B{b}, \B{d}} \chi_\B{b}^{*2}\chi_\B{d}\chi^*_{\B{d}-2\B{b}},
\end{align}
which establishes that $(7)=(3)=(4)$.
    
    \item $\pi\in \{(12)(34), (14)(23)\}$ 

These permutations ask us to consider
\begin{align}
 \sum_{\textbf{a}}  \trace(D_{\textbf{a}}\Pi_{\mathcal{E}} D_{\textbf{a}}^{\dag}\Pi_{\mathcal{E}} ) \trace (D_{\textbf{a}}\Pi_{\mathcal{E}} D_{\textbf{a}}^{\dag}\Pi_{\mathcal{E}} ) 	.
\end{align}
We have
\begin{align}
	\trace(D_{\textbf{a}}\Pi_{\mathcal{E}} D_{\textbf{a}}^{\dag}\Pi_{\mathcal{E}} )&= \trace\left(\sum_\B{b} \chi_\B{b}\omega^{[\B{a},\B{b}]}D_{\B{b}} \sum_\B{c}\chi_\B{c}^*D_\B{c}^\dagger \right)\\
	&=\sum_{\B{b},\B{c}} \chi_\B{b} \chi_\B{c}^*\omega^{[\B{a}, \B{b}]}d\delta_{\B{b}, \B{c}}=d\sum_{\B{b}} |\chi_\B{b}|^2\omega^{[\B{a},\B{b}]},
\end{align}
and so
\begin{align}
 &\sum_{\textbf{a}}  \trace(D_{\textbf{a}}\Pi_{\mathcal{E}} D_{\textbf{a}}^{\dag}\Pi_{\mathcal{E}} ) \trace (D_{\textbf{a}}\Pi_{\mathcal{E}} D_{\textbf{a}}^{\dag}\Pi_{\mathcal{E}} )\\
 &=(d^n)^2\sum_\B{a}\sum_{\B{b}} |\chi_\B{b}|^2\omega^{[\B{a},\B{b}]}\sum_{\B{c}} |\chi_\B{c}|^2\omega^{[\B{a},\B{c}]}\\
 &=(d^n)^2 \sum_{\B{b},\B{c}}|\chi_\B{b}|^2|\chi_\B{c}|^2\sum_\B{a}\omega^{[\B{a}, \B{b+c}]}\\
  &=(d^n)^2 \sum_{\B{b},\B{c}}|\chi_\B{b}|^2|\chi_\B{c}|^2 d^2\delta_{\B{b}, -\B{c}}\\
  &=(d^n)^4 \sum_{\B{b}}|\chi_\B{b}|^2|\chi_{\B{b}}^*|^2\\
  &= (d^n)^4 \sum_{\B{b}}|\chi_\B{b}|^4,
\end{align}
which shows that $(8)=(1)$.

  \item $\pi=(13)(24)$  
    
    This final permutation yields
\begin{align}
\sum_{\textbf{a}}  \trace(D_{\textbf{a}} \Pi_{\mathcal{E}}D_{\textbf{a}} \Pi_{\mathcal{E}})\trace(D_{\textbf{a}}^{\dag}\Pi_{\mathcal{E}}D_{\textbf{a}}^{\dag}\Pi_{\mathcal{E}})	.
\end{align}
Recalling that $
D_\B{a}^\dagger\Pi_{\mathcal{E}} D_\B{a}^\dagger=\sum_\B{b} \chi_\B{b}D_{\B{b}-2\B{a}}$, we have
\begin{align}
\trace\left(\sum_\B{b} \chi_\B{b}D_{\B{b}-2\B{a}}\sum_\B{c}\chi^*_\B{c}D_\B{c}^\dagger \right) &= \sum_{\B{b}, \B{c}} \chi_\B{b}\chi_\B{c}^* d\delta_{\B{b}-2\B{a},\B{c}}\\
&=d^n\sum_{\B{b}}\chi_\B{b}\chi^*_{\B{b}-2\B{a}},
\end{align}
so that
\begin{align}
\sum_{\textbf{a}}  \trace(D_{\textbf{a}} \Pi_{\mathcal{E}}D_{\textbf{a}} \Pi_{\mathcal{E}})\trace(D_{\textbf{a}}^{\dag}\Pi_{\mathcal{E}}D_{\textbf{a}}^{\dag}\Pi_{\mathcal{E}})&= (d^n)^2 	\sum_{\B{a}}\left|\sum_{\B{b}}\chi_\B{b}\chi^*_{\B{b}-2\B{a}}\right|^2.
\end{align}
\end{enumerate}
Putting it all together, we find
\begin{align}
\trace\left(Q\Pi_{\mathcal{E}}^{\otimes 4}\Pi_{\text{sym}^4} \right) &= \frac{1}{24}\Bigg\{3 (d^n)^2\sum_{\B{a}} |\chi_\B{a}|^4 + 6 d^n\sum_{\B{a}, \B{b}}  |\chi_\B{a}|^2  |\chi_\B{b}|^2 \omega^{[\B{a}, \B{b}]}+6 d^n\sum_{\textbf{a}, \textbf{b}}\chi_{\textbf{a}}^{2}\chi_\textbf{b}\chi_{\textbf{b}+2\textbf{a}}^*\\
&+8 \sum_{\B{a}, \B{b}, \B{c}}\chi_\B{a}\chi_\B{b}\chi_\B{c}\chi^*_{\B{a}+\B{b}+\B{c}}\tau^{[\B{a}, \B{b}] -[\B{a}, \B{c}]-[\B{c}, \B{b}]}+ \sum_{\B{a}}\left|\sum_{\B{b}}\chi_\B{b}\chi^*_{\B{b}-2\B{a}}\right|^2\Bigg\}\nonumber
\end{align}
where $\chi_\B{a} = \frac{1}{d^n}\trace(D_\B{a}^\dagger \Pi_\mathcal{E})$.

\section{\texorpdfstring{Proof of \cref{zerogap}}{Proof of Theorem 2}}\label{proofzerogap}
In this section we will state and prove \cref{zerogap}.
\begin{theorem}
    The average stabilizer gap of a stabilizer codespace with  isotropic set $\mathcal{S}$ and trivial group homomorphism $f(\B{a})=0\,\,\,\forall \B{a}\in \mathcal{S}$, with corresponding isometry $\mathcal{E}_\mathcal{S}$, is 
   \begin{align}
	\Delta M(\mathcal{E}_\mathcal{S})&= \frac{\alpha d_B - |A_\mathcal{S}|d_S}{d_B(d_S+1)(d_S+2)(d_S+3)} && \alpha = \begin{cases}1, & d_S \text{ odd (multiqudit)} \\
	4, & d_S \text{ even (multiqudit)}\\
    d_S^2, &d_S=2^m \text{ (multiqubit)}. \end{cases}
\end{align}
where $|A_\mathcal{S}|$ is the cardinality of the set
\begin{align}
A_\mathcal{S} = \Big\{ \B{a} \in \mathbb{Z}_d^{2n}\ | \ 2\B{a}\in \mathcal{S}\quad and\quad 	\B{a}\in S^\perp \Big\}.
\end{align}
with $S^\perp\equiv\{ \B{a}\in \mathbb{Z}_d^{2n} \,\,| \,\,\forall \B{b}\in \mathcal{S}:[\B{a}, \B{b}]=0\}$
\end{theorem}

\begin{proof}[\cref{zerogap}]
Let $\mathcal{S}$ be a totally isotropic set in the sense of \cref{stab_subspaces}, and let
\begin{align}
    \Pi_{\mathcal{E}_{\mathcal{S}}} = \frac{1}{|\mathcal{S}|}\sum_{\B{a}\in \mathcal{S}} D_{\B{a}}
\end{align}
be the projector onto the corresponding stabilizer codespace. We will consider the characteristic function $\chi_{\B{a}}$ of the projector $\Pi_{\mathcal{E}_S}$. As we observed in \cref{ase_gap}, from
\begin{align}
D_\B{a}=D_{\B{x}+d\B{y}}=\begin{cases}
        D_{\B{x}} &\text{$d$ odd}\\ (-1)^{[\B{x},\B{y}]}D_{\B{x}}\quad &\text{$d$ even,}
        \end{cases}
\end{align}
where $\B{a} = \B{x}+ d\B{y}$ for $\B{x} = \B{a} \mod d$ and $\B{y} = \frac{1}{d}(\B{a}-\B{x})$, we may derive a symmetry of the characteristic function,
\begin{align}
\chi_{\B{a}}=\chi_{\B{x}+d\B{y}}=\begin{cases}
        \chi_{\B{x}} &\text{$d$ odd}\\ (-1)^{[\B{x},\B{y}]}\chi_{\B{x}}\quad &\text{$d$ even},
        \end{cases}	
\end{align}
where we note 
\begin{align}
[\B{x}, \B{y}]  &=\Big[\B{x},  \frac{1}{d}(\B{a}-\B{x})\Big]=\frac{1}{d}\Big[\B{x}, \B{a}\Big].
\end{align}
The characteristic function of our projector $\Pi_{\mathcal{E}_\mathcal{S}}$ is then
\begin{align}
\chi_{\B{a}}=\frac{1}{|\mathcal{S}|}\begin{cases}
       \delta_{\B{x} \in S} &\text{$d$ odd}\\ (-1)^{\frac{1}{d}[\B{x},\B{a}]}\delta_{\B{x} \in \mathcal{S}}\quad &\text{$d$ even}.
        \end{cases}	
\end{align}
 We may now evaluate $\trace\left(Q\Pi_{\mathcal{E}_S}^{\otimes 4}\Pi_{\text{sym}^4} \right)$ using the expression derived in Appendix \ref{App_Average}. Clearly for the first two terms,
\begin{align}
\sum_{\B{a}} |\chi_{\B{a}}|^4 && 	\sum_{\B{a}, \B{b}}  |\chi_{\B{a}}|^2  |\chi_{\B{b}}|^2 \omega^{[\B{a}, \B{b}]},
\end{align}
the phase factor makes no difference. Thus for both $d$ odd and even,
\begin{align}
	\sum_{\B{a}} |\chi_{\B{a}}|^4&= \frac{1}{|\mathcal{S}|^4}\sum_{\B{a}} \delta_{\B{a}\in \mathcal{S}}=\frac{1}{|\mathcal{S}|^3}
\end{align}
\begin{align}
	\sum_{\B{a}, \B{b}}  |\chi_{\B{a}}|^2  |\chi_{\B{b}}|^2 \omega^{[\B{a}, \B{b}]}&=\frac{1}{|\mathcal{S}|^4}\sum_{\B{a}, \B{b}}\delta_{\B{a}\in S}\delta_{\B{b}\in \mathcal{S}} =\frac{1}{|S|^2}.
\end{align}
For the next two terms,
\begin{align}	
\sum_{\B{a}, \B{b}}\chi_{\B{a}}^{2}\chi_{\B{b}}\chi_{\B{b}+2\B{a}}^* &&  \sum_{\B{a}, \B{b}, \B{c}}\chi_{\B{a}}\chi_{\B{b}}\chi_{\B{c}}\chi^*_{\B{a}+\B{b}+\B{c}}\tau^{[\B{a}, \B{b}]-[\B{a},\B{c}]-[\B{c}, \B{b}]},
\end{align}
we argue as follows. If $d$ is even, we have
\begin{align}
\sum_{\B{a}, \B{b}}\chi_{\B{a}}^{2}\chi_{\B{b}}\chi_{\B{b}+2\B{a}}^* =\frac{1}{|\mathcal{S}|^4}\sum_{\B{a}, \B{b}}\delta_{\B{a}\in \mathcal{S}}\delta_{\B{b}\in \mathcal{S}}(-1)^{[\B{x}, \B{y}]}\delta_{\B{x}\in \mathcal{S}}, 	
\end{align}
where  $\B{c}=\B{b}+2\B{a}=\B{x}+d\B{y}$. Notice that the only terms that survive are those for which $\B{a}, \B{b}, \B{x}\in \mathcal{S}$, and so in particular $[\B{x}, \B{a}]=[\B{x}, \B{b}]=0$.  Thus
\begin{align}
[\B{x}, \B{y}]=\frac{1}{d}\Big[\B{x},\B{c}\Big]=\frac{1}{d}\Big[\B{x}, \B{b}+2\B{a}\Big]=0.
\end{align}
We conclude that 
\begin{align}
\sum_{\B{a}, \B{b}}\chi_{\B{a}}^{2}\chi_{\B{b}}\chi_{\B{b}+2\B{a}}^*  &=\frac{1}{|\mathcal{S}|^4}\sum_{\B{a}, \B{b}}\delta_{\B{a}\in \mathcal{S}}\delta_{\B{b}\in \mathcal{S}}\delta_{\B{x}\in \mathcal{S}}=\frac{1}{|\mathcal{S}|^2},
\end{align}
where $\B{x} = \B{b} + 2\B{a} \mod d$, and this is the same in the case that $d$ is odd.
The same argument applies to the next term, so that
\begin{align}
	 \sum_{\B{a}, \B{b}, \B{c}}\chi_{\B{a}}\chi_{\B{b}}\chi_{\B{c}}\chi^*_{\B{a}+\B{b}+\B{c}}\tau^{[\B{a}, \B{b}]-[\B{a},\B{c}]-[\B{c}, \B{b}]}&=\frac{1}{|\mathcal{S}|^4} \sum_{\B{a}, \B{b}, \B{c}}\delta_{\B{a}\in \mathcal{S}}\delta_{\B{b}\in \mathcal{S}}\delta_{\B{c}\in \mathcal{S}}\delta_{x\in \mathcal{S}}=\frac{1}{|\mathcal{S}|},
\end{align}
where $\B{x} = \B{a} + \B{b} + \B{c} \mod d$, and we've used the fact that since $\B{a}, \B{b}, \B{c}\in \mathcal{S}$, their symplectic products all vanish. We conclude in particular that the first four terms are insensitive to whether $d$ is odd or even.

Finally, we must tackle the fifth term
\begin{align}
	\sum_{\B{a}}\left|\sum_{\B{b}}\chi_{\B{b}}\chi^*_{\B{b}-2\B{a}}\right|^2 &= \frac{1}{|\mathcal{S}|^4}\sum_{\B{a}}\left|\sum_{\B{b}}\delta_{\B{b}\in \mathcal{S}}(-1)^{[\B{x}, \B{y}]}\delta_{\B{x}\in \mathcal{S}}\right|^2,
\end{align}	
where $\B{b}-2\B{a}=\B{x} + d\B{y}$. A term in the inner sum is non-zero if $\B{b},\B{x}\in \mathcal{S}$. We have
\begin{align}
\B{x}\equiv \B{b}-2\B{a}\pmod d \Longleftrightarrow \B{b}-\B{x} \equiv 2\B{a}	\pmod d,
\end{align}
and since $\mathcal{S}$ is closed under addition, $\B{b}-\B{x}\in \mathcal{S}$: in fact, any element of $\mathcal{S}$ can be written this way. Thus equivalently, we require $\B{b}, 2\B{a}\in \mathcal{S}$. Moreover, because 
\begin{align}
[\B{x}, \B{y}] &= \frac{1}{d}[\B{x}, \B{b}-2\B{a}]=\frac{2}{d}[\B{a}, \B{x}]	
\end{align}
(where we've used the fact that $[\B{x}, \B{b}]=0$), we have
\begin{align}
(-1)^{[\B{x}, \B{y}]}&=(e^{i \pi })^{\frac{2}{d}[\B{a}, \B{x}]}=\omega^{[\B{a},\B{x}]}\omega^{[\B{a},(\B{b}-2\B{a})\text{ mod  }d]}=\omega^{[\B{a},\B{b}-2\B{a}]}=\omega^{[\B{a}, \B{b}]},
\end{align}
since $\omega^{a b}=\omega^{a(x + dy)}=\omega^{a (b \text{ mod } d)}(\omega^d)^{ay}=\omega^{a (b \text{ mod } d)}$. Thus
\begin{align}
	\sum_{\B{a}}\left|\sum_{\B{b}}\chi_{\B{b}}\chi^*_{\B{b}-2\B{a}}\right|^2 &= \frac{1}{|\mathcal{S}|^4}\sum_{\B{a}}\left|\sum_{\B{b}\in \mathcal{S}}\omega^{[\B{a}, \B{b}]}\right|^2\delta_{2\B{a}\in \mathcal{S}}.
\end{align}	
At the same time, because it is a character sum of a finite abelian group,
\begin{align}
\sum_{\B{b}\in \mathcal{S}}\omega^{[\B{a}, \B{b}]}=\begin{cases} |\mathcal{S}|, & \text{if } \forall \B{b}\in \mathcal{S}:[\B{a}, \B{b}]=0\\
0, & \text{otherwise}.\end{cases}	
\end{align}
The former condition is that $\B{a}\in \mathcal{S}^\perp$, and so we have
\begin{align}
\sum_{\B{a}}\left|\sum_{\B{b}}\chi_{\B{b}}\chi^*_{\B{b}-2\B{a}}\right|^2 &= \frac{1}{|\mathcal{S}|^4}\sum_{\B{a}}|\mathcal{S}|^2\delta_{2\B{a}\in \mathcal{S}}\delta_{\B{a}\in \mathcal{S}^\perp}=\frac{|A_\mathcal{S}|}{|\mathcal{S}|^2},
\end{align}
where we define the set $A_\mathcal{S}$ as
\begin{align}
A_\mathcal{S} = \Big\{ \B{a} \in \mathbb{Z}_d^{2n}\ | \ 2\B{a}\in \mathcal{S}, 	 \forall \B{b}\in \mathcal{S}:[\B{a}, \B{b}]=0\Big\}.
\end{align}
Putting everything together, we find
\begin{align}
\trace\left(Q_B\Pi_{\mathcal{E}}^{\otimes 4}\Pi_{\text{sym}^4} \right) &= \frac{1}{24}\Bigg[\frac{3 d_B^2}{|\mathcal{S}|^3} + \frac{6 d_B}{|\mathcal{S}|^2}+ \frac{6 d_B}{|\mathcal{S}|^2}+ \frac{8}{|\mathcal{S}|}+ \frac{|A_\mathcal{S}|}{|\mathcal{S}|^2}\Bigg]\\
&= \frac{d_S\Big(|A_\mathcal{S}| d_S + d_B(3d_S^2+12d_S+8)\Big)}{24 d_B^2},
\end{align}
so that the average SE of the subspace with respect to the big space is  
\begin{align}
\mathbb{E}_U[M_B(\mathcal{E}_\mathcal{S}(U\psi U^\dagger))] 
&=1-d_B \binom{d_S+3}{4}^{-1} \trace\left(Q_B\Pi_{\mathcal{E}_\mathcal{S}}^{\otimes 4}\Pi_{\text{sym}^4}^B \right)\\
&=1-\frac{d_B(3d_S^2 + 12 d_S + 8) + |A_\mathcal{S}|d_S}{d_B(d_S+1)(d_S+2)(d_S+3)}\\
&= 1-\frac{3(d_S+2)}{(d_S+1)(d_S+3)} - \frac{\frac{d_S}{d_B}|A_\mathcal{S}|-4}{(d_S+1)(d_S+2)(d_S+3)},
\end{align}
where we've written the last in view of the fact that
\begin{align}
\mathbb{E}_U[M_S(U\psi U^\dagger)] &=\begin{dcases}
    1-\frac3{d_S+2}, &  d_S \text{ odd (multiqudit)} \\
    1-\frac{3(d_S+2)}{(d_S+1)(d_S+3)},   & d_S \text{ even (multiqudit)} \\
    1- \frac{4}{d_S+3}, & d_S=2^m \text{ (multiqubit)}.\\
\end{dcases}
\end{align}
\end{proof}

\section{\texorpdfstring{$|A_\mathcal{S}|$}{|As|} for a \texorpdfstring{$\mathbb{Z}_d$}{Zd}-invariant subspace}\label{gauge}

In \cref{exgauge}, we considered the set $\mathcal{S}$ given by all symplectic vectors of the form $(x,0,\dots,x,0) \in \mathbb{Z}_d^{2n}$ for $x\in \mathbb{Z}_d$. Our goal in this appendix is to determine the cardinality of the set $A_\mathcal{S} = \{ \mathbf{a} \in \mathbb{Z}_d^{2n} \mid 2\mathbf{a} \in \mathcal{S}, \mathbf{a} \in \mathcal{S}^\perp \}$.

On the one hand, it is immediate to see that $\mathcal{S}$ is generated by the single element $\B{g}=(1,0,\dots,1,0)\in\mathbb{Z}_d^{2n}$. On the other hand, we can write a generic element $\B{a} \in \mathbb{Z}_d^{2n}$ in the form $(x_1,y_1,\dots,x_n,y_n)$. For $\B{a}$ to be in $A_\mathcal{S}$, it must satisfy two conditions. The first condition is $2\mathbf{a}=(2x_1, 2y_1, \dots, 2x_n, 2y_n) \in \mathcal{S}$.  This tells us a) $2y_i \equiv 0 \pmod{d}$ for all $i=1, \dots, n$, and b) all the odd-position components must be equal, i.e., $2x_1 \equiv 2x_2 \equiv \dots \equiv 2x_n \pmod{d}$. The second condition, $\mathbf{a} \in \mathcal{S}^\perp$, is equivalent to requiring $[\mathbf{g}, \mathbf{a}]=0$. To see this, note that the symplectic product of a generic element $\B{b}\in \mathcal{S}$ is $[\B{b}, \B{a}]=\lambda[\B{g}, \B{a}]$ for some $\lambda \in \mathbb{Z}_d$. If $\B{a}\in \mathcal{S}^\perp$, then $[\B{g}, \B{a}]=0$, and so $[\B{b}, \B{a}]=0$. Conversely, if $ [\B{g},\B{a}]=0$, then $\lambda[\B{g}, \B{a}]=[\B{b}, \B{a}]=0$ for all $\B{b}\in \mathcal{S}$, and so $\B{a}\in \mathcal{S}^\perp$. Next we observe that
\begin{equation}
[\mathbf{g}, \mathbf{a}] = \sum_{i=1}^{n} (1 \cdot y_i - 0 \cdot x_i) = \sum_{i=1}^{n} y_i,
\end{equation}
which shows that the second condition is equivalent to $\sum_{i=1}^{n} y_i \equiv 0 \pmod{d}$. We can now count the number of vectors $\mathbf{a}$ that satisfy these simultaneous conditions, a task which depends on the parity of $d$. 

First, let's consider the case where $d$ is odd. In this case, $\text{gcd}(2, d) = 1$. The condition $2y_i \equiv 0 \pmod{d}$ implies that $y_i \equiv 0$ for all $i$. This choice, $(y_1, \dots, y_n)=(0, \dots, 0)$, trivially satisfies the sum condition $\sum y_i \equiv 0 \pmod d$. Thus, there is only one possibility for the $y$-components. For the $x$-components, the condition $2x_1 \equiv \dots \equiv 2x_n \pmod d$, given that $d$ is odd, simplifies to $x_1 \equiv \dots \equiv x_n \pmod d$. We can choose $x_1$ in $d$ ways, which then determines all other $x_i$. This gives $d$ possibilities for the $x$-components. The total cardinality is the product of the choices, so for odd $d$, $|A_S| = d \times 1 = d$.

Next, we consider the case where $d$ is even. Here, $\text{gcd}(2, d) = 2$. The condition $2y_i \equiv 0 \pmod{d}$ has two solutions for each $y_i$: $y_i=0$ and $y_i=d/2$. The sum condition $\sum y_i \equiv 0 \pmod{d}$ then requires that the number of $y_i$ equal to $d/2$ must be even. The number of ways to form a vector of length $n$ with components from $\{0, d/2\}$ with an even number of non-zero entries is $2^{n-1}$. This is the number of choices for the $y$-components. For the $x$-components, we begin by choosing $x_1$ in $d$ ways. The condition $2x_i \equiv 2x_1 \pmod{d}$ has $\text{gcd}(2, d)=2$ solutions for each of the $n-1$ variables $x_2, \dots, x_n$. This gives $d \times 2^{n-1}$ possibilities for the $x$-components. The total cardinality is the product $(d \cdot 2^{n-1}) \times (2^{n-1}) = d \cdot 2^{2n-2}=d\cdot4^{n-1}$.

\end{document}